%% file: arxiv.tex
\pdfoutput=1

\documentclass{article}

\usepackage{microtype}
\usepackage{amssymb}
\usepackage{graphicx}

\usepackage{hyperref}
\usepackage{subcaption}
\include{preamble_def}

\usepackage{multicol}

\begin{document}
\title{Minmax Tree Facility Location and Sink Evacuation with Dynamic Confluent Flows}

\author{Di Chen \\ HKUST \and Mordecai Golin \\ HKUST}


\maketitle

\begin{abstract}
Let $G=(V,E)$ be a graph modelling  a building or road network in which  edges  have-both  travel times (lengths) and   {\em capacities} associated with them. An edge's
capacity is the number of people that can enter that edge in a unit of time.  
 In emergencies, people evacuate towards the exits.  If too many people try to evacuate through the same edge, {\em congestion} builds up and slows down the evacuation.

Graphs with both lengths and capacities are known as   {\em Dynamic Flow networks}.
An {\em evacuation plan} for $G$  consists of a choice of exit locations  and a partition of the people at the vertices into groups,  with each group evacuating to the same exit.   The {\em evacuation time} of  a plan is the time it takes  until  the last person  evacuates.  The {\em $k$-sink evacuation problem} is to provide an evacuation plan with $k$ exit locations  that minimizes the evacuation time. It is known that this problem is NP-Hard for general graphs but no polynomial time algorithm was previously known even for the case of $G$ a tree. This paper presents an $O(n k^2 \log^5 n)$  algorithm for the $k$-sink evacuation problem on trees. Our algorithms also apply to a more general class of problems, which we call minmax tree facility location.
\end{abstract}

\input{mjgintro}
\input{intro}
\input{overview}

\input{bounded_cost}
\input{full_problem}
\input{conclusion}

\bibliography{arxiv}

\newpage
\appendix

\input{appendix}

\end{document}

%% file: preamble_def.tex
\usepackage{verbatim}
\usepackage{enumitem}
\usepackage{amsmath, amsthm, amssymb, amsfonts}
\usepackage{graphicx}
\usepackage{algpseudocode}
\usepackage{algorithmicx}
\usepackage{algorithm}

\usepackage{epsfig}
\usepackage{epstopdf}
\newtheorem{theorem}{Theorem}
\newtheorem{lemma}[theorem]{Lemma}
\newtheorem{corollary}[theorem]{Corollary}
\newtheorem{definition}{Definition}

%% file: mjgintro.tex
\section{Introduction}
\label{sec: mjgintro}

{\em Dynamic flow networks} model movement of items on a graph.
Each vertex $v$ is assigned   some initial set of supplies $w_v$.  Supplies flow across edges.  Each edge $e$ has a length -- the time required to traverse  it  -- and  a   capacity $c_e$,  which  limits the rate of the flow of supplies into the edge in one time unit.  If all edges have the same capacity $c_e=c$ the network is said to have {\em uniform capacity}.   As supplies move around the graph,
 {\em congestion} can occur as supplies back up at a vertex, increasing the time necessary to send a flow.
 
 Dynamic flow networks  were introduced by Ford and Fulkerson 
in \cite{Ford1958a}   and have since been extensively used and  analyzed.   There are essentially two basic types of problems, with many variants of each.
These are  the {\em Max Flow over Time (MFOT)} problem   of how much flow can be moved  (between specified vertices)  in a given time $T$ 
and the {\em Quickest Flow Problem (QFP)} of how quickly a given  $W$ units of flow can be moved.  Good surveys of the area and applications can be found in
\cite{Skutella2009,Aronson1989,Fleischer2007,Pascoal2006}.

One variant  of the QFP  that is of interest is the transshipment  problem, e.g., \cite{Hoppe2000b}, in which  the graph has several sources and sinks, with the original supplies being the sources and each sink having a specified demand.  The problem is to find the minimum time required to satisfy all of the demands.  \cite{Hoppe2000b} designed the first polynomial time algorithm for that problem, for the case of integral travel times.

Variants of QFP Dynamic flow problems can also model \cite{Higashikawa2014} {\em evacuation problems}.  In these,  vertex supplies are people in a building(s) and the problem is to find a routing strategy (evacuation plan) that evacuates all of them to specified sinks (exits)  in minimum time.  Solving this using  (integral) dynamic  flows, would assign each person an evacuation path with, possibly, two people at the same vertex travelling radically different paths.

A slightly modified version of the problem, the one addressed here, is for the plan to assign to each vertex $v$ exactly one   exit or {\em evacuation edge}, i.e., a sign stating ``this way out''. All people starting or arriving at $v$ must evacuate through that edge. After traversing the edge they  follow the evacuation edge at the arrival vertex. They continue following  the unique evacuation edges  until reaching a sink,  where they exit.
The simpler optimization problem is, given the sinks,   to determine  a plan minimizing the total time needed to evacuate everyone.  A more complicated version is, given $k$, to find the (vertex) locations of the $k$ sinks/exits {\em and} associated evacuation plan that  together minimizes the evacuation time. This is the {\em $k$-sink location problem}.  

Flows with the property that all  flows entering a vertex  leave along  the same edge are known as {\em confluent}\footnote{Confluent flows occur naturally in problems other than evacuations, e.g.,  packet forwarding and railway scheduling \cite{Dressler2010b}.}; even  in the static case constructing  an optimal confluent flow is known to be very difficult. i.e.,  if P $\not=$ NP, then it is even impossible to construct a constant-factor approximate optimal confluent flow in polynomial time on a general graph  \cite{Chen2007,Dressler2010b,Chen2006,Shepherd2015}.
 
 Note that if the capacities are  ``large enough'' then no congestion can occur and every person follows the shortest path to some exit with the cost of the plan being the length of the  maximum shortest path.  This is exactly the $k$-center problem on graphs which is already known to be NP-Hard \cite[ND50]{garey1979computers}. Unlike  $k$-center, which is polynomial-time solvable for fixed $k$,  Kamiyama {\em et al.}~\cite{Kamiyama} proves 
 by reduction to {\em Partition},   that, even for fixed   $k=1$   finding the min-time evacuation protocol is still NP-Hard for general graphs

The only solvable known case for general $k$ is for $G$ a path. For paths with uniform capacities,  \cite{higashikawa2015multiple} gives an $O(kn)$ algorithm.\footnote{This is generalized to  the general capacity path to  $O(k n \log^2 n)$ in the unpublished \cite{ArumugamAGS15}.}

When $G$ is a tree
the $1$-sink location problem can be solved    \cite{Mamada2006} in $O(n \log^2 n)$ time.   This can be reduced  \cite{Higashikawa2014} down to $O(n \log n)$ for the uniform capacity version, i.e., all all the $c_e$ are identical. If the {\em locations of the $k$ sinks are given as input},
\cite{Mamada2005} gives an $O(n^2 k \log^2  n)$ algorithm for finding the minimum time evacuation protocol. i.e., a partitioning of the tree into subtrees that evacuate to each sink. The best previous known time for solving the $k$-sink location problem was $O(n (c \log n)^{k+1})$, where $c$ is some constant \cite{Mamada2005a}.


\subsection{Our contributions}
This paper gives the first polynomial time algorithm for solving the $k$-sink location problem on trees. Our result  uses the  $O(n \log^2 n)$ algorithm of \cite{Mamada2005}, for calculating the evacuation time of a tree given the location of a sink,  as an oracle.

\begin{theorem}
	The $k$-sink evacuation problem can be solved in time $O(n k^2 \log^5 n)$.
	\label{theorem:kSinkRunTime}
\end{theorem}


It is instructive to compare our approach to Frederickson's  \cite{frederickson1991parametric}  $O(n)$ algorithm for solving the $k$-center problem on trees, which was built from the  following two ingredients.

\begin{enumerate}
\item An $O(n)$ time previously known algorithm for checking {\em feasibility}, i.e., given $\alpha >0$, testing whether a  $k$-center solution with cost $\le \alpha$ {\em  exists}
\item A clever {\em parametric search} method to filter the $O(n^2)$ pairwise distances between nodes, one of which is the optimal cost, via the feasibility test.
\end{enumerate}

Section \ref{Section: Bounded Cost}, is devoted to constructing a first polynomial time feasibility test for $k$-sink evacuation on trees. It starts with a simple version that makes a ploynomial number of oracle calls and then is extensively refined so as to make only $O(k \log n)$ (amortized) calls.

On the other hand, there is no small set of easily defined cost values known to contain the optimal solution.  We sidestep this issue by doing parametric searching {\em within} our feasibility testing algorithm, Section \ref{Section: Full Problem}, which leads to Theorem \ref{theorem:kSinkRunTime}. 

As  a side result,  a slight modification to our algorithm allows improving, for almost all $k$,  the best previously  known algorithm for solving the problem when the $k$-sink locations are already given,  from  $O(n^2 k \log^2 n)$ \cite{Mamada2005}  down to  $O(n k^2 \log^4 n)$.



\section{Formal definition of the sink evacuation problem}

Let $G=(V,E)$ be an undirected graph. Each  edge $e=(u,v)$   has a travel time $\tau_e$;  flow leaving  $u$ at time $t=t_0$ arrives at $v$ at time $t=t_0 + \tau_e.$ 
Each edge also has  a
{\em capacity} $c_e \ge 0$. This  restricts at most  $c_e$ units of resource to enter edge $e$  in one unit of time. For our version of the problem we restrict $c$ to be integral; the capacity can then be visualized as the {\em width} of the edge with only $c_e$ people being allowed to travel in parallel along  the edge.

Consider $w_u$ people waiting at vertex $u$ at time $t=0$ to travel the edge $e=(u,v).$ Only $c_e$ people  enter the edge in one time unit, so the items travel in $\lceil w_u/c_e\rceil$ packets, each of size $c_e$,  except possibly for the last one.  
The first packet enters $e$ at time $t=0,$ the second at time $t=1$, etc..
The first packet therefore reaches $v$ at time $t=\tau_e$ time, the second at $t=\tau_e +1$ and the last one at time $t=\tau_e + \lceil w_u/c_e \rceil -1$. Figure 
\ref{fig:evacs}(a)  illustrates this process.  In the diagram  people  get moved along $(u,v)$ in groups of size at most  6.  If $w_u = 20$, there are 4 groups; the first one reaches $v$ at time  $t=10$, the second  at time $t=11$,  the third at $t=12$ and the last  one (with only 2 people)  at $t=13.$  


\begin{figure}[t]
	\begin{subfigure}[h]{0.35\textwidth}
		\includegraphics[width=\textwidth]{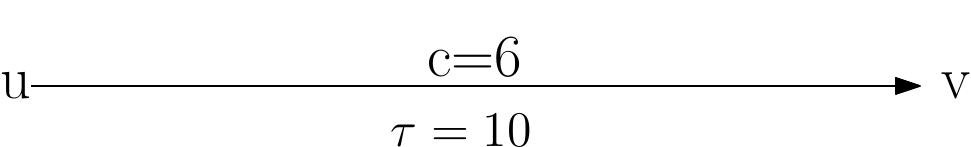}
		\caption{}
	\end{subfigure}%
	\hfill
	\begin{subfigure}[h]{0.6\textwidth}
		\includegraphics[width=\textwidth]{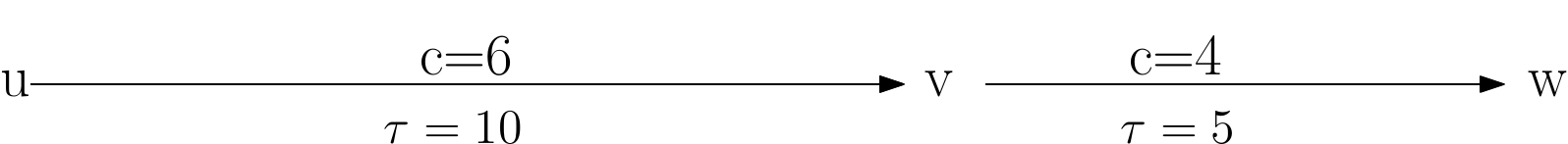}
		\caption{}
	\end{subfigure}
	\caption
	{\small In (a),  if $w_u=20$ the last person leaving $u$ arrives at $v$ at time $t=13$.
		In (b) Assume people at $u,v$ are all evacuating to $w$ and $w_u=20$ and $w_v > 0$.
		The first person from $u$ arrives at $v$ at time $t=11$.
		If $w_v \le 40$ all of the people on $v$ enter $(v,w)$  {\em before or at } $t=10$, so  there will  be no congestion  when the first people from $u$ arrive at $v$ and  they just sail through $v$ without stopping.  Calculation shows that the last people from $u$ reach $w$ at time $t=20.$
		On the other hand, if $w_v > 40,$ some people who started at $v$ will still be waiting at $v$ when the first people from $u$ arrive there.
		In this case, there is congestion and the people from $u$ will have to wait.  A little calculation shows that,  after waiting, the last person from $u$ will finally arrive at $w$ at time $ 14 + \lfloor (20+w_v)/4\rfloor$.
	}
	\label{fig:evacs}
\end{figure}

Now suppose that items are travelling  along a path $\ldots u \rightarrow v \rightarrow w \rightarrow \ldots$ where $e_1 = (u,v)$ and $e_2=(v,w)$. Items arriving at $v$ can't enter $e_2$ until the items already there have left.  This waiting causes {\em congestion} which is one of the major complications involved in constructing good evacuation paths. Figure 
\ref{fig:evacs}(b) 
illustrates how congestion can build up.


As another example, consider Figure  \ref{fig:newevac}(a) with every node evacuating to $w$.
When the first people from $u_1$ arrive at $v$,  some of the original people still remain there,  leading to congestion.  Calculation shows that the last people from $u_1$ leave $v$ at time 4 so when the first  people from $u_2$ arrive at $v$ at time 5, no one is waiting at $v$.  But, when the first people from $u_3$ arrive at $v$ some people from $u_2$ are waiting  there, causing congestion. After that, people arrive   from both $u_2$ and $u_3$ at the same time, with many having to wait.  The last person finally reaches $w$ at time 15, so the evacuation protocol takes time 15.

\begin{figure}
	\centering
		\begin{subfigure}[h]{0.5\textwidth}
			\includegraphics[width=\textwidth]{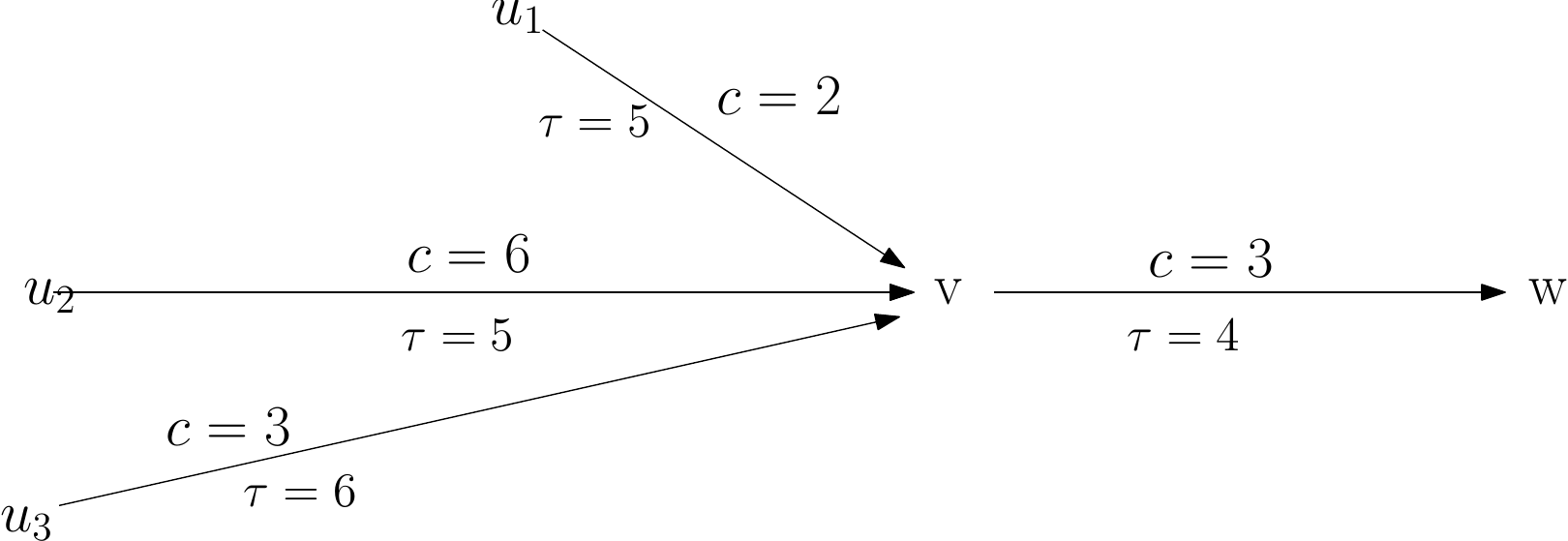}
			\caption{}
		\end{subfigure}%
		\begin{subfigure}[h]{0.5\textwidth}
			\includegraphics[width=\textwidth]{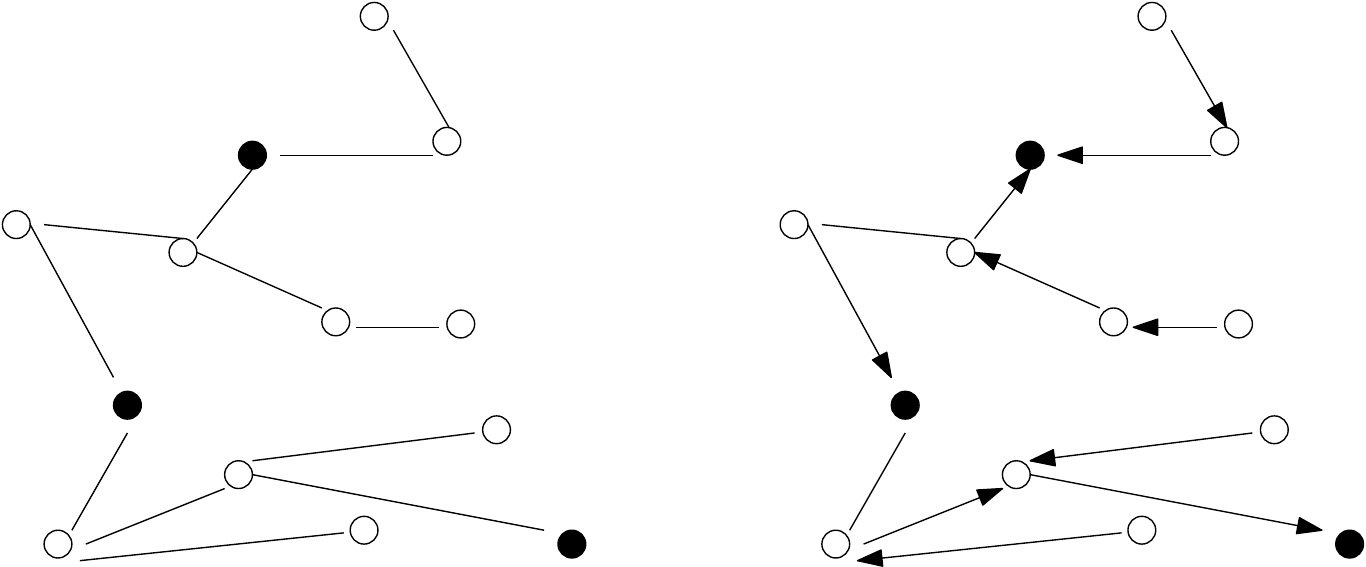}
			\caption{}
		\end{subfigure}
		\caption{ \small   (a) evacuation problem with 4 vertices and sink at $w$. If  $w_v=8$, $w_{u_1}=4$, $w_{u_2}=10$, $w_{u_3}=11$ and sink at $w$,  the last person arrives $w$ at time 15. In (b) the
			left figure is a tree with the $k=3$ black vertices as sinks.  The right figure  provides an evacuation plan.   Each non-sink vertex $v$ has exactly one outgoing edge and, following the directed edges from each such $v$ leads to a sink.}
		\label{fig:newevac}
\end{figure}


Given a graph $G$, distinguish a subset $S \subseteq V$ with $|S|=k$  as sinks (exits).
An evacuation plan provides, for each vertex $v\not\in S$,  the unique edge along which all people starting at or passing through $v$ evacuate.  Furthermore, starting at any $v$ and following  the edges will lead from $v$ to one of the $S$  (if $v \in S$, people at $v$ evacuate immediately through the exit at $v$). Figure \ref{fig:newevac}(b)  provides an example.


Note that the evacuation plan defines a confluent flow.  The evacuation edges form a directed forest;  the sink of each tree in the forest is  one of the designated sinks in $S.$. 

Given the evacuation plan and the values $w_v$ specifying the initial number of people located at each node,  one can calculate, for each vertex, the  time (with congestion) it takes for all of its people to evacuate.  The maximum of this over all $v$  is the minimum time required to required to evacuate  {\em all}  people  to some exit  using the rules above. Call this the {\em cost for $S$ associated with the evacuation plan}.  The {\em cost} for $S$ will be the minimum cost over all evacuation plans using that set $S$ as sinks.

The {\em $k$-sink location} problem is to find a subset $S$ of size $k$ with minimum cost. Recall that \cite{Mamada2005}  provides an $O(n \log^2 n)$ problem for solving this problem if for tree  $G$ with  $k=1$.  We will use this algorithm as an oracle for solving the general $k$-location problem on trees.

Given the hardness results, it is unlikely one can produce an efficient algorithm for general graphs, but our algorithms can serve as fast subroutines for exhaustive search or heuristic methods commonly employed in practice.

%% file: intro.tex
%

\subsection{General problem formulation}

The input to our algorithm(s) will be  a tree $T_{\mathrm{in}} = (V_{\mathrm{in}},E_{\mathrm{in}})$,  and a positive integer $k$. Let $n = |V_{\mathrm{in}}| = |E_{\mathrm{in}}| + 1$. Our goal will be to find a subset  $S \subseteq V_{\mathrm{in}}$ with cardinality at most $k$ that can minimize cost  $F(S)$.  This will essentially involve partitioning the $V_{\mathrm{in}}$  into $\le k$ subtrees that minimizes their individual max costs.


 We note that  our algorithms will not explicitly deal with the mechanics of evacuation calculations.  Instead they will solve the location problem for any {\em monotone min-max cost $F(S)$}. 
   We introduce this level of abstraction because using the clean  properties of monotone min-max cost  functions makes the algorithms easier to formulate and understand.

%

\subsubsection{Monotone min-max cost.}
We now extract the properties of $F(S)$ that we will use.  All of these will be consistent with evacuation time.
Let $\Lambda[S]$ be the set of all partitions of $V_{\mathrm{in}}$ such that, for each $\mathcal{P} \in \Lambda[S]$, and each $P \in \mathcal{P}$, we have $|S \cap P| = 1$, and $P$ induces a connected component in $T_{\mathrm{in}}$.

Intuitively, each $\mathcal{P} \in \Lambda[S]$ is a partition of $T_{\mathrm{in}}$ into $|S|$ subtrees, such that each subtree includes exactly $1$ element in $S$. For any $P \subseteq V_{\mathrm{in}}$ s.t. $ | S \cap P | = 1$, we denote by $\langle S \cap P \rangle$ the unique node $v \in S \cap P$. We say that nodes in $P$ are assigned to the \emph{sink} $\langle S \cap P \rangle$.

Now we define an \emph{atomic cost function} $f : 2^{V_{\mathrm{in}}} \times V_{\mathrm{in}} \rightarrow [0,+\infty]$. In the context of facility location problems, given $P \subseteq V_{\mathrm{in}}$ and $|S \cap P| = 1$, $f(P,\langle S \cap P \rangle)$ can be interpreted as the cost for sink $\langle S \cap P \rangle$ to serve the set of nodes $P$. The definition of $f$ involves some natural constraints on cost functions for facility location on trees, which are given as follows. 
\begin{enumerate}
\item For $U \subseteq V_{\mathrm{in}}$, $v \in V_{\mathrm{in}}$,
  \begin{itemize}
  \item if $v \notin U$ then $f(U,v) = +\infty$;
  \item if $U$ does not induce a connected component, $f(U,v) = +\infty$.
  \item if $U = \{v\}$, then $f(U,v) = 0$.
  \end{itemize}
\item (Set monotonicity) If $v \in U_1 \subseteq U_2 \subseteq V_{\mathrm{in}}$, then $f(U_1,v) \leq f(U_2,v)$. i.e. the cost tends to increase when a sink has to serve additional nodes.

\item (Path monotonicity) Let $u \in U$ and let $v \notin U$ be a neighbor of $u$ in $T_{\mathrm{in}}$. Then $f(U \cup \{ v \},v) \geq f(U,u)$. Intuitively, this means when we move a sink away from $U$, the cost for the sink to serve $U$ tends to increase.

\item (Max composition) Let $T = (U,E')$ be a connected component in $G$, and $v \in U$. Let $\mathcal{F} = \{ T_1,...,T_l \}$ be the forest created by removing $v$ from $T$, and the respective vertices of each tree in $\mathcal{F}$ be $U_1,...,U_l$. Then $f(U,v) = \max_{1\leq i\leq l} f(U_i \cup \{ v\},v)$.
\end{enumerate}
Note  that we have only defined a cost function over one single set and one single sink. This can then be extended to a function $F_S : \Lambda[S] \rightarrow [0,+\infty]$ by setting, for $\mathcal{P} \in \Lambda[S]$:
\begin{equation}
F_S(\mathcal{P}) = \max_{P \in \mathcal{P}} f(P,\langle S \cap P \rangle)
\label{eq:MaxCompositionB}
\end{equation}
In other words, given a partition $\mathcal{P}$, the total cost for sinks $S$ to serve all partitioned blocks is the maximum of the cost to serve each block. It will be  cumbersome to discuss explicit partitioning, so we will informally denote it by saying that a node $u \in V_{\mathrm{in}}$ is {\em assigned} to a sink $s \in S$. Then, given sinks $S$, we partition $G$ in a way that the total cost is minimized, giving the cost function as:
\begin{equation}
 F(S) = \min_{\mathcal{P} \in \Lambda(S)} F_S(\mathcal{P})
\end{equation}
We call such cost function $F$ minmax monotone. See the appendix for an illustration of (\ref{eq:MaxCompositionB}). $k$-center and sink evacuation will fit into this framework, as well as variations involving node capacities, uniform edge capacity, or confluent unsplittable flows.
 Our main  problem will be to find an $S$ which minimizes $F(S)$ over all $|S| = k$.   

Our algorithms are designed to make calls directly to an oracle $\mathcal{A}$ that computes $f(U,v)$ given any $U$ that induces a connected component of $T_{\mathrm{in}}$ and any $v \in U$. In general such a polynomial time oracle must exist for the problem to even be in NP. 


%% file: overview.tex
\section{Overview}

In the rest of the paper, we will describe two versions of our algorithms. In either version we require a feasibility test, which solves a simplified, \emph{bounded cost} version of the problem.

\begin{table}[H]
	\centering
	\begin{tabular}{ |r|p{0.7\textwidth}| }
		\hline
		{\bfseries Problem} & Bounded cost minmax $k$-sink \\
		\hline
		Input & Tree $T_{\mathrm{in}}=(V_{\mathrm{in}},E_{\mathrm{in}})$, $k \geq 1$, $\mathcal{T} \geq 0$\\
		Output & $S_{\mathrm{out}} \subseteq V_{\mathrm{in}}$ and $\mathcal{P_{\mathrm{out}}} \in \Lambda[S_{\mathrm{out}}]$ s.t. $|S_{\mathrm{out}}| \leq k$ and $F_{S_{\mathrm{out}}}(\mathcal{P_{\mathrm{out}}}) \leq \mathcal{T}$.\ \ \ 
		If such  a  $(S_{\mathrm{out}}, \mathcal{P_{\mathrm{out}}})$  pair does not exist, output `No'. \\
		\hline
	\end{tabular}
\end{table}

We will use an algorithm solving  this problem  as a subroutine for solving the full problem; we measure the time complexity by the number of calls to the oracle $\mathcal{A}$. The fastest runtime we can obtain is given as follows.

\begin{theorem}
	If $\mathcal{A}$ runs in time $t_\mathcal{A}(n)$, the bounded cost minmax $k$-sink problem can be solved in time $O(k t_\mathcal{A}(n) \log n)$ if $t_\mathbb{A}(n)$ is at least linear time. \label{theorem:FastBC}
\end{theorem}
If $\mathcal{A}$ is sublinear, we can replace it with a linear time oracle to get $O(k n  \log n)$. We will establish several important ingredients that leads to $O(n)$ calls, i.e. a time complexity fo $O(t_\mathcal{A}(n) n)$; the same ingredients will be used in the more complicated algorithm that gives Theorem ~\ref{theorem:FastBC}.

%% file: bounded_cost.tex
\section{Bounded cost $k$-sink (feasibility test)}
\label{Section: Bounded Cost}

\begin{definition}
A {\em feasible configuration} is a set of sinks $S \subseteq V$ with a partition $\mathcal{P} \in \Lambda(S)$ where $F_S(\mathcal{P}) \leq \mathcal{T}$; $S$ is also separately called a {\em feasible sink placement}, and $\mathcal{P}$ is {\em a partition  witnessing the feasibility of $S$}. An optimal feasible configuration is a {\em  feasible sink placement $S^* \subseteq V$ with minimum cardinality}; we write $k^* := |S^*|$.
\end{definition}

If $k^* > k$ then the algorithm returns `No'. Otherwise, it returns a feasible configuration $(S_{\mathrm{out}},\mathcal{P}_\mathrm{out})$ such that $|S_{\mathrm{out}}| \leq k$.

\begin{definition}
Suppose $U$ induces a subtree of $T_\mathrm{in}$ and $S \subseteq U$. We say $U$ can be {\em served by $S$} if, for some partition $\mathcal{P}$ of $U$, for each $P \in \mathcal{P}$ there exists $s \in S$ such that $f(P,s) \leq \mathcal{T}$.
\end{definition}

\begin{definition}
  Let $U$ be the nodes of a connected component of $G$ and $v\in V$ (not necessarily in $U$). We say that $v$ {\em supports} $U$ if one of the following holds:  
  \begin{itemize}
    \item If $v \in U$, then $f(U,v) \leq \mathcal{T}$.
    \item If $v \notin U$, let $\Pi$ be the set of nodes on the path from $v$ to $U$. Then $f(U \cup \{v\} \cup \Pi, v) \leq \mathcal{T}$.
  \end{itemize}
\end{definition}

If $U$ can be served by $S$, then any node in $U$ is supported by some $s \in S$. The converse is not necessarily true.

\subsection{Greedy construction}

Our algorithms  greedily build $S_{\mathrm{out}}$ and $\mathcal{P}_{\mathrm{out}}$ on-the-fly, making irrevocable  decisions on what should be in the output. $S_{\mathrm{out}}$ is initialized to be empty. In each step, we add elements to $S_{\mathrm{out}}$ but never remove them, and once $|S_{\mathrm{out}}| > k$ we immediately terminate with a `No'. If, at termination, $|S_{\mathrm{out}}| \leq k$, we output $S_{\mathrm{out}}$.

Similarly, $\mathcal{P}_{\mathrm{out}}$ is initially empty, and the algorithm performs irrevocable  updates to $\mathcal{P}_{\mathrm{out}}$ while running. An update to $\mathcal{P}_{\mathrm{out}}$ is a \emph{commit}. When set $P_\mathrm{new} \subseteq V_{\mathrm{in}}$ is committed it is associated with some some sink in $S_{\mathrm{out}}$ (which might have to be added to $S_{\mathrm{out}}$ at the same time).  If $P_\mathrm{new}$  shares its  sink with an existing block $P \in \mathcal{P}_\mathrm{out}$, we merge $P_\mathrm{new}$ into $P$.  Another way to view this operation is that either a new sink is added, or unassigned nodes are assigned to an existing sink.

\begin{algorithm}[h]
	\begin{algorithmic}[1]
		\State Given $\mathcal{P}_{\mathrm{out}}$
		
		\Procedure{Commit}{$P_\mathrm{new} \subseteq V_{\mathrm{in}}$}
		    \If{$|P_{new} \cap S_\mathrm{out}| = 1$}
		      \If{$\exists P \in \mathcal{P}_{\mathrm{out}}$ s.t. $P_\mathrm{new} \cap P \neq \emptyset$}
		      \State $\mathcal{P}_{\mathrm{out}} := \mathcal{P}_{\mathrm{out}} \cup \{ P \cup P_\mathrm{new} \} - \{ P \}$
		      \Else
		      \State $\mathcal{P}_{\mathrm{out}} := \mathcal{P}_{\mathrm{out}} \cup \{ P_\mathrm{new} \}$
		      \EndIf
		    \EndIf
		\EndProcedure

	\end{algorithmic}
	\caption{Committing block}
	\label{alg:CommitBlock}
\end{algorithm}

In essence, we avoid backtracking so that $S_\mathrm{out}$ does not lose elements, and blocks added to $\mathcal{P}_\mathrm{out}$ can only grow. For this to work, we must require, throughout the algorithm:

\begin{enumerate}[label=(\textrm{C}\arabic{*})]
 \item An optimal feasible sink placement $S^*$ exists where $S_{\mathrm{out}} \subseteq S^*$. \label{invariant:test}
 \item For any $P \in \mathcal{P}_{\mathrm{out}}$ there exists a unique $s \in S_{\mathrm{out}}$ such that $|P \cap S| = 1$, and $f(P,s) \leq \mathcal{T}$.
\end{enumerate}
 
 Additionally, $\mathcal{P}_{\mathrm{out}}$ will be a partition of $V_\mathrm{in}$ upon termination with `yes'. When these conditions all hold, then
 $|\mathcal{P}_\mathrm{out}| \leq k$  and $(S_{\mathrm{out}}, \mathcal{P}_\mathrm{out})$ is feasible and  output by the algorithm.




\subsubsection{A separation argument.}

As the algorithm progresses, it removes nodes from the remaining graph (the {\em working tree}), simplifying the combinatorial structure. Roughly speaking, a sink can be removed if we can identify all nodes it has to serve; upon removing the sink, all nodes it serves can also be removed from the tree. We will need the definitions below:

\begin{definition}[Self sufficiency and $T_{-v}(u)$]
A subtree $T' = (V',E')$ of $T_{\mathrm{in}}$ is {\em self-sufficient} if $V'$ can be served by $S_{\mathrm{out}} \cap V'$. 
	
Given a tree $T = (V,E)$, consider an internal node $v \in V$
and one of its neighbors $u \in V$.  Removing $v$ from $T$ leaves  a forest $\mathcal{F}_{-v}$ of disjoint subtrees of $T$.

There is a unique tree $T' = (V',E') \in \mathcal{F}_{-v}$ such that $u \in V'$, denoted by $T_{-v}(u) = (V_{-v}(u), E_{-v}(u))$. The concept of self sufficiency 
is introduced for subtree of this form.
\end{definition}

If $T_{-v}(u)$ is self-sufficient, and $u$ is a sink, 
%
 there is no need to add any other sinks to $T_{-v}(u)$, also no node oustside $T_{-v}(u)$ will be routed to any sink in $T_{-v}(u)$ other than $u$. 
This means all nodes in $V_{-v}(u)$ except $u$ can be removed from consideration; 
a more formal statement of this fact  is given as follows.

Given subtrees $T_1 = (V_1,E_1)$ and $T_2 = (V_2,E_2)$ of $T_{\mathrm{in}}$, we denote by $T_1 \backslash T_2$ the graph induced by $V_1 \backslash V_2$.

\begin{lemma}
	Given $s \in S_{\mathrm{out}}$, suppose $u \in V_{\mathrm{in}}$ is a neighbor of $s$ in $T_{\mathrm{in}}$. Consider the subtree $T' = (V',E')$ induced by vertices $(V_\mathrm{out})_{-s}(u) \cup \{s\}$, and suppose $T'$ is self-sufficient. Then the following are equivalent.
	
	
	\begin{enumerate}
		\item There exists $S_1^* \subseteq V_\mathrm{in}$ such that $S_\mathrm{out} \subseteq S_1^*$, $F(S_1^*) \leq \mathcal{T}$, and $|S_1^*| \leq k$.
		\item There exists $S_2^* \subseteq (V_\mathrm{in} \backslash V') \cup \{s\}$ such that $(S_\mathrm{out} \backslash V') \cup \{s\} \subseteq S_2^*$, $F(S_2^*) \leq \mathcal{T}$ when restricted to $(V_\mathrm{in} \backslash V') \cup \{s\}$, and $|S_2^*| + |S_{\mathrm{out}} \cap V'| - 1\leq k$
	\end{enumerate}
	
	In other words, we can ignore all nodes (including sinks) that are in $V' \backslash \{ s \}$, and solve the problem on the subtree induced by $(V_{\mathrm{in}} \backslash V') \cup \{ s \}$.
	\label{lemma:Separation}
\end{lemma}

\begin{proof}
	(i) $\Rightarrow$ (ii): Set $S^*_2 := S^*_1 \cap (V_{\mathrm{in}} \backslash V') \cup \{s\}$, and let $\mathcal{P}^*$ be the partition that witnesses the feasibility of $S^*_1$. Then for any $P \in \mathcal{P}^*$ such that $P \cap V' \neq \emptyset$, we have by max composition and (\ref{eq:MaxCompositionB}) that $\mathcal{T} \geq f(P, \langle P \cap S^*_1 \rangle) \geq f(P \cap V', \langle P \cap S^*_2 \rangle)$, so $V'$ can be served by $S^*_2$.
	
	(ii) $\Rightarrow$ (i): Let $\mathcal{P}_0$ be the partition of $T'$ that witnesses the self-sufficiency of $T'$, and let $\mathcal{P}_1$ be the partition of the subtree induced by $T'' := (V_\mathrm{in} \backslash V') \cup \{s\}$ that witnesses the self-sufficiency of $T''$. Take $S_1^* := S_2^* \cup (S_{\mathrm{out}} \cap V')$. Then $F(S_1^*) = \max_{\langle S_1^* \cap P \rangle: P \in \mathcal{P}_0 \cup \mathcal{P}_1} f(P,s)$ by max composition and (\ref{eq:MaxCompositionB}), which is at most $\leq \mathcal{T}$ by assumption. Also, because $s \in S_2^* \cap (S_{\mathrm{out}} \cap V')$, we know that $|S_1^*| \leq |S_2^*| + |S_{\mathrm{out}} \cap V'| - 1 \leq k$. \
\end{proof}

Throughout the algorithm, we maintain a `working' tree $T = (V,E)$ 
as well as a working set of sinks $S = S_\mathrm{out} \cap V$. 
Initially,   $T=T_{\mathrm{in}}$.  As the algorithm progresses, $T$ is maintained to be a subtree of $T_{\mathrm{in}}$ by peeling off self-sufficient subtrees. Lemma \ref{lemma:Separation} ensures that  solving the bounded problem on $T$ is equivalent to solving the bounded problem on $T_{\mathrm{in}}$.

To use Lemma \ref{lemma:Separation}, we enforce that  sink $s$  is added to $S_{\mathrm{out}}$ and $S$ only when, for some neighbor $u$ of $s$, the tree induced by $V_{-s}(u) \cup \{s\}$ is self-sufficient with respective to the sink set $S \cup \{s\}$. This permits removing $V_{-s}(u)$ from $T$ after adding $s$ to $S_{\mathrm{out}}$ and $S$. So in the algorithm
we can assume that sinks exist only at the leaves of the working tree $T$.

\subsection{Subroutine: Peaking Criterion}

%

We now  describe a convenient mechanism that allows us to greedily add sinks.

\begin{definition}[Peaking criterion]
 Given $T = (V,E)$, the ordered pair of points $(u,v) \in V \times V$ satisfies the {\em peaking criterion} (abbreviated {\em PC}) if and only if $u$ and $v$ are neighbors,  $T_{-v}(u)$ has no sink, and finally $f(V_{-v}(u), u) \leq \mathcal{T}$ but $f(V_{-v}(u) \cup \{ v \} , v) > \mathcal{T}$.
\end{definition}

\begin{lemma}
 Let $S$ be a feasible sink placement for $T$, and let $u,v \in V$ be neighbors. If $(u,v)$ satisfies the peaking criterion, then $S' := (S \backslash V_{-v}(u)) \cup \{ u \}$ is also a feasible sink placement. In particular, if $S$ is an optimal feasible sink placement, then so is $S'$.
 \label{lemma:PeakingCriterionPutSink}
\end{lemma}
\begin{proof}[Sketch Proof]
 By path monotonicity no sink outside $V_{-v}(u)$ can support $V_{-v}(u)$, so the only choice is to put a sink in $V_{-v}(u)$. By $(iii)$, the best place to put the sink is then $u$.
\end{proof}

\begin{proof}[Full Proof Lemma \ref{lemma:PeakingCriterionPutSink}]
	By path monotonicity, $(iv)$ implies that none of the nodes in $V \backslash V_{-v}(u)$ support $V_{-v}(u)$. On the other hand, by $(iii)$ and set monotonicity, $(S \backslash V_{-v}(u)) \cup \{ u \} )$ is a feasible sink placement.
	
	Moreover, $F(S) \leq \mathcal{T}$ implies $|S \cap V'| \geq 1$, so $|(S \backslash V') \cup \{ u \}| \leq |S|$. In other words, we know that any feasible configuration needs to place a sink in $T_{-v}(u)$ configuration, and no sink needs to exist as a descendent of of $u$ in $T$. This in turn implies that all sinks in $T_{-v}(u)$ can be replaced with one single sink.
\end{proof}

If $(u,v)$ satisfies the peaking criterion, we can immediately place  a sink at $u$  and then  commit $V_{-v}(u)$. 
 The following demonstrates that, whenever $S = \emptyset$, at least $1$ sink can be found using the peaking criterion, unless a single node can $s \in V$ support the entire graph.

\begin{lemma}
 Suppose for some $v,u$, $f(V_{-v}(u) \cup \{v\}, v) > \mathcal{T}$, and $S \cap V_{-v}(u) = \emptyset$. Then there exists a pair of nodes $s,t \in V_{-v}(u) \cup \{v\}$ such that $(s,t)$ satisfies the peaking criterion.
 \label{lemma:PCRecurseInto}
\end{lemma}
\begin{proof}[Proof of Lemma \ref{lemma:PCRecurseInto}]
	Assume for a contradiction that no pair $(s,t) \in V_{-v}(u) \times V_{-v}(u)$ satisfies the peaking criterion.
	
	This implies $f(V_{-v}(u), u) > \mathcal{T}$, which in turn implies $|V_{-v}(u)| \geq 2$, because $f(\{u\},u) = 0\leq \mathcal{T}$. Then max composition imply there exists $\eta_0 \in V_{-v}(u), \eta_0 \neq u$ that is a neighbor of $u$ such that $f(V_{-u}(\eta_0) \cup \{u\}, u) > \mathcal{T}$.
	
	Applying this repeatedly will generate an endless sequence of distinct nodes $\eta_0,\eta_1,\eta_2,...$ such that $\eta_i$ is a neighbor of $\eta_{i+1}$ but $f(V_{-\eta_i}(\eta_{i+1}) \cup \{\eta_i\}, \eta_i) > \mathcal{T}$, which is impossible because $T$ is finite. 
	%
\end{proof}

\begin{corollary}
 Given $S = \emptyset$, either one of the following occurs:
 \begin{enumerate}
  \item For any $s \in V$ we have $f(V,s) \leq \mathcal{T}$, or
  \item There exist a pair of nodes $u,v \in V$ that satisfies the peaking criterion.
 \end{enumerate}
 \label{corollary:PCStoppingCondition}
\end{corollary}
\begin{proof}[Proof of Corollary \ref{corollary:PCStoppingCondition}]
	Suppose (i) does not hold. Then for all $v \in V$ we have that $f(V,v) > \mathcal{T}$; by max composition of $f$, this means for every $v \in V$ there exists a neighbor $u$ of $v$ such that $f(V_{-v}(u) \cup \{v\}, v) > \mathcal{T}$. Then Lemma \ref{lemma:PCRecurseInto} implies (ii).
	
	Now suppose instead (ii) does not hold. Take $v \in V$. By Lemma \ref{lemma:PCRecurseInto}, for every $u \in V$ that is a neighbor of $v$, we know that $f(V_{-v}(u) \cup \{v\}, v) \leq \mathcal{T}$; this in turn implies, by max composition, $f(V,v) \leq \mathcal{T}$. As $v$ was taken arbitrarily, we have (i). \
\end{proof}

%
%

At stages where it is applicable, for each ordered pair $(u,v)$ that satisfies the peaking criterion we place a sink at $u$ and remove nodes in $V_{-v}(u)$. If instead the first case of the above corollary occurs, we can add an arbitrary $s \in V$ to $S$ and $S_{\mathrm{out}}$ and terminate.

\subsection{Hub tree}

Corollary \ref{corollary:PCStoppingCondition}  provides  two ways to add sinks to $S_\mathrm{out}$.
We do not add sinks  any other way. But merely applying this principle does not find all sinks. We therefore  in Section \ref{subsec: Subroutine: Dual Peaking Criterion} introduce a new process that complements the peaking criterion. First,  we introduce the hub tree, which has convenient properties that arise from applying the peaking criterion.

\begin{definition}[Hubs]
 Let $L \subseteq V$ be the leaves of the rooted tree $T=(V,E)$. 
 Let  $S \subseteq L$, be a set of sinks, with no sink in $V \backslash L$. Let $H(S) \subseteq V$ be the set of lowest common ancestors of all pairs of sinks in $T$. The nodes in $H(S)$ are the {\em hubs} associated with $S$. 
 
 The {\em hub tree}  $T_{H(S)} = (V_{H(S)}, E_{H(S)})$ is the subgraph of $T$ that includes all vertices and edges along all possible simple paths among nodes $H(S) \cup S$. (See illustration in appendix)
\end{definition}

\begin{definition}[Outstanding branches]
Given $T = (V,E)$ and $S$, we say that a node $w \in V$ branches out to $\eta$ if $\eta$ is a neighbor of $w$ in $T$ that does not exist in $V_{H(S)}$. The subtree $T' := T_{-w}(\eta)$ is called an {\em outstanding branch}; we say that $T'$ is {\em attached} to $w$.
\end{definition}

\begin{definition}[Bulk path]
 Given two distinct $u,v \in V_{H(S)}$, the {\em bulk path} $\mathrm{BP}(u,v)$ is the union of nodes along the unique path $\Pi$ between $u,v$ (inclusive), along with all the nodes in all outstanding branches that are attached to any node in $\Pi$.
\end{definition}

A crucial property arises after an exhaustive application of the peaking criterion.

\begin{definition}[RC-viable]
 Given $T$ and sinks $S$, we say that $T$ is {\em RC-viable} if:
 \begin{enumerate}
  \item all sinks $S$ occur at the leaves of $T$
  \item if $T' = (V',E')$ is an outstanding branch attached to $w \in V_{H(S)}$, then $f(V' \cup \{ w \},w) \leq \mathcal{T}$
 \end{enumerate}
\end{definition}

\begin{lemma}
 Given $T = (V,E)$ and sinks $S$, where $S$ is a subset of leaves of $T$. Suppose no ordered pair $(u,v) \in V \times V$ satisfy the peaking criterion. Then $T$ is RC-viable.
 \label{lemma:RCViability}
\end{lemma}

\begin{proof}[Proof of 	Lemma \ref{lemma:RCViability}]
	
	Let $T' = (V',E')$ be an arbitrary outstanding branch attached to some node $w \in V_{H(S)}$. It suffices to show that $f(V' \cup \{ w \},w) \leq \mathcal{T}$.
	
	As $T'$ is an outstanding branch, we know $V' \cap S = \emptyset$. By assumption no pair $u,v $ in $V' \cup \{ w\}$ would satisfy the peaking criterion, so from Lemma \ref{lemma:PCRecurseInto} we know that $f(V' \cup \{ w \},w) \leq \mathcal{T}$. \
\end{proof}

Now when $T$ is RC-viable w.r.t. $S$, there is no need to place sinks within outstanding branches; this is because if an outstanding branch is attached to a node $w$, then a sink at $w$ can already serve the the entire outstanding branch. 

\subsection{Subroutine: Reaching Criterion}
\label{subsec: Subroutine: Dual Peaking Criterion}

The peaking criterion is a way to add sinks to $T$ and remove certain nodes from $T$. On the other hand, the reaching criterion is a way to \emph{remove} sinks from $T$ and $S$, while keeping them in $S_{\mathrm{out}}$. Roughly speaking, the reaching criterion finalizes all nodes that should be assigned to certain sinks, and then removes all these nodes from consideration.

Given $T$ and $S$, we say a node $v \in E_{H(S)}$ \emph{can evacuate} to $s \in S$ if $f(\mathrm{BP}(v,s),s) \leq \mathcal{T}$; when such $s \in S$ exists for $v$, we say that $v$ \emph{can evacuate}. Given this we can formulate an `opposite' to the peaking criterion, which allows us to remove nodes, including sinks, from $T$.


\begin{definition}[Reaching criterion]
 Given $T = (V,E)$ and a set of sinks $S$, placed at the leaves of $T$. Let  $T$ be RC-viable with respect to $S$ and  $(u,v) \in V \times V$ be an ordered pair of nodes.
 $u,v$ satisfy the {\em reaching criterion (RC)} if and only if they are neighbors in $T$, and $T_{-v}(u)$ is self-sufficient while the tree induced by $\mathrm{BP}(v,u) \cup V_{-v}(u)$ is not. 
\end{definition}

\begin{theorem}
 Suppose $T = (V,E)$ is RC-viable with respect to $S \subseteq V$. If $u,v \in V$ satisfies the reaching criterion, then we can remove $T_{-v}(u)$ from $T$, and also commit all blocks in the partitioning of $T_{-v}(u)$ that witnesses the self-sufficiency of $T_{-v}(u)$. By definition, $T_{-v}(u)$ includes at least one sink from $S$.
 \label{theorem:RCTrimTree}
\end{theorem}
\begin{proof}[Proof of Theorem \ref{theorem:RCTrimTree}]
	As $T_{-v}(u)$ is self sufficient, no additional sink has to be placed in it. By RC-viability, no sink has to be placed in outstanding branches, either.
	
	Formally, this means there exists an optimal feasible sink configuration $(S^*,\mathcal{P}^*)$ where $S^*$ contains no node in any outstanding branch attached to $v$; furthermore, $S^* \cap V_{-v}(u) = S \cap V_{-v}(u)$. 
	
	Now suppose for a contradiction that there exists $P \in \mathcal{P}^*$ such that $u,v \in P$. Because a block in $\mathcal{P}^*$ has to induce a connected component, and there is no sink in the outstanding branches attached to $v$, we know $\mathrm{BP}(v,u)$ has to be served by sinks in $V_{-v}(u)$, which violates the assumption that the tree induced by $\mathrm{BP}(v,u) \cup V_{-v}(u)$ is not self-sufficient.
	
	So $v$ and $u$ can not be assigned to the same block in $\mathcal{P}^*$. This in turn implies that none of the blocks in $\mathcal{P}^*$ can span nodes in both $V \backslash V_{-v}(u)$ and $V_{-v}(u)$, because each block has to induce a connected component in $T_\mathrm{in}$. Thus $T_{-v}(u)$ is no longer relevant and can be safely removed. \
\end{proof}

 
After removing $T_{-v}(u)$ by the reaching criterion, 
we need to run the peaking criterion again on $T$, in order to preserve RC-viability.

%
\subsubsection{Testing for self-sufficiency.}
In order to make use of the reaching criterion, we require efficient tests for self-sufficiency. Note that \cite{Mamada2005} readily gives such test albeit at a higher time complexity. In our algorithm, we perform self-sufficiency tests on a rooted subtree $T'$ only if it satisfies some special conditions, allowing us to exploit RC-viability and reuse past computations. By our arrangements, when such $T'$ passes our test we know it demonstrates a stronger form of self-sufficiency.

\begin{definition}[Recursive self-sufficiency]
 Given a rooted subtree $T' = (V',E')$ of $T$, $V' \cap S \neq \emptyset$, we say that $T'$ is {\em recursively self-sufficient} if for all $u \in V_{H(S)} \cap V'$, the subtree of $T'$ rooted at $u$ is self-sufficient.
\end{definition}

A bottom-up approach can be used to test for recursive self-sufficiency, which in turn implies `plain' self-sufficiency. 

\begin{lemma}
 Given a RC-viable rooted subtree $T' = (V',E')$ of $T$, $V' \cap S \neq \emptyset$, where $v$ is the root. 
 Suppose there exists a child $u$ of $v$ in $V_{H(S)} \cap V'$ such that $T_{-v}(u)$ is recursively self-sufficient, and there is a sink $s \in S \cap V_{-v}(u)$ such that $u$ can evacuate to $s$.
 
 Then $BP(v,s) \cup V_{-v}(u)$ is recursively self-sufficient. If, additionally, for every child $u'$ of $v$ in $V_{H(S)} \cap V'$, $T_{-v}(u')$ is recursively self-sufficient, then $T'$ is recursively self-sufficient.
 \label{lemma:RecursiveSS}
\end{lemma}
\begin{proof}[Proof of Lemma \ref{lemma:RecursiveSS}]
	Suppose we assign $v$ to sink $s$. This implies that all nodes in $BP(v,s)$ are assigned to $s$. On the other hand, consider the graph induced by $V_{-v}(u) \backslash BP(v,s)$. For any node $v' \in (V_{H(S)} \cap V_{-v}(u)) \backslash BP(v,s)$, the subtree of $T'$ rooted at $v'$ is self sufficient, because $T_{-v}(u)$ is recursively self-sufficient.
	
	Additionally, if every child $u'$ of $v$ in $V_{H(S)} \cap V'$ is such that $T_{-v}(u')$ is recursively self-sufficient, by a similar argument we know that $T'$ is recursively self-sufficient. \
\end{proof}

We say that $s$ is a witness to Lemma \ref{lemma:RecursiveSS} for $T'$ and $v$; we store this witness, as well as the witness for every subtree of $T'$ rooted at some $v \in V'$. From the proof of Lemma \ref{lemma:RecursiveSS} one can see it is easy to retrieve a partition $\mathcal{P}'$ of $T'$ that witnesses the self-sufficiency of $T'$, in $O(|V'|)$ time. See Algorithm \ref{alg:FindPartitionRecursiveSS} in appendix.


%
%
%

For this to be useful, note that only recursive self-sufficiency will be relevant. 
When a RC-viable tree is self-sufficient but not recursively self-sufficient, if we process bottom-up, we can always cut off part of the tree using the reaching criterion, so that the remainder is recursively self-sufficient. This is demonstrated in the detailed algorithm.
\subsection{Combining the Pieces}

The main ingredients of our algorithm are the peaking and reaching criteria along with ideas to test self-sufficiency. We use the peaking criterion to add sinks to $T$, and then the reaching criterion to remove sinks and nodes from $T$, until either $T$ is empty or $T$ can be served by a single sink. In the following we describe a full algorithm that makes use of these ideas. 

\subsubsection{Simpler, iterative approach (`Tree Climbing')}

Essentially, in this algorithm we iteratively check and apply the two peaking criteria bottom-up from the leaves. We do not specify a root here; the root can be arbitrary, and changed whenever necessary.  As we go up from the leaves, for each pair $(u,v)$ that forms an edge of the tree, we would call the oracle $\mathcal{A}$ for $f(V_{-v}(u), u)$, $f(V_{-v}(u) \cup \{v\}, v)$ or $f(\mathrm{BP}(v,s),s)$ for some sink $s$, and apply either the peaking criterion or the reaching criterion. By design RC is checked whenever the tree is RC-viable, and PC is checked whenever the tree is not RC-viable, and we do not need to test both on the same pair $(u,v)$. 

\begin{lemma}
 The bounded-cost tree-climbing (Algorithm \ref{alg:BoundedCostFull}) makes $O(n)$ calls to $\mathcal{A}$.
 \label{lemma:BoundedCostCalls}
\end{lemma}
\begin{proof}
 We only make $O(1)$ calls to evaluate $f(\cdot, \cdot)$ for each pair $(u,v) \in E_\mathrm{in}$.
\end{proof}

After seeing the iterative approach, it is easier to understand the more advanced algorithm, which uses divide-and-conquer and binary search to replace the iterative processes.

\subsubsection{Peaking criterion by recursion.}

Macroscopically, we replace plain iteration  with a fully recursive process. We do this once in the beginning, as well as every time we remove a sink. Overall the algorithm makes $O(k \log n)$ `amortized' calls to the oracle. Recall that the main purpose of the peaking criterion is to place sinks and make the tree $T$ RC-viable.

%
%
%
%

\paragraph{A localized view.}  We start with a more intuitive, localized view of the recursion.
We evaluate $f(\cdot,\cdot)$ on sets of nodes of the form $V_{-v}(u)$ or $V_{-v}(u) \cup \{v\}$. If $f(V_{-v}(u),u) \leq \mathcal{T}$ then we mark all nodes in $V_{-v}(u)$. Sometimes we also mark the node $v$, if all but one of its neighbors are marked.

Over the course of the algorithm, we are given a node $v \in V$ (along with other information including $\mathtt{Marked}_{\mathrm{PC}}$), and for each neighbor $u$ of $v$ we decide whether to evaluate $a_u := f(V_{-v}(u) \cup \{v\}, v)$. As a basic principle to save costs, we do not wish to call the oracle if all nodes in $V_{-v}(u) \cup \{v\}$ are marked, or if $V_{-v}(u) \cup \{v\}$ contains a sink.

When we do get $a_u \leq \mathcal{T}$, we put all nodes in $V_{-v}(u)$ into $\mathtt{Marked}_{\mathrm{PC}}$. Moreover, if at least $|N(v) - 1|$ neighbors of $v$ are in $\mathtt{Marked}_{\mathrm{PC}}$, and by this time $v \notin S_{\mathrm{out}}$, we also put $v$ into $\mathtt{Marked}_{\mathrm{PC}}$. This part is the same in tree-climbing, and maintains \emph{an important invariant} regarding $\mathtt{Marked}_\mathrm{PC}$: if $u$ is marked but a neighbor $v$ is not, then all nodes in $V_{-v}(u)$ are marked, and $f(V_{-v}(u) \cup \{v\},v) \leq \mathcal{T}$.

On the other hand, if in fact we find that $a_u > \mathcal{T}$, we would wish to recurse into $T_{-v}(u)$, because one sink must be placed in it. Now we return to a more global view.

\paragraph{A global view.}

To maintain RC-viability we need to apply the oracle on various parts of $T$. In the iterative algorithm, this process is extremely repetitive. Now we wish to segregate different sets of nodes on the tree, so the oracle is only applied to separate parts.


\begin{definition}[Compartments and Boundaries]
	Let $T'=(V',E')$ be a subtree of $T=(V,E)$. The boundary $\delta T'$ of $T'$ is the set of all nodes in $T'$ that is a neighbor of some node in $V \backslash V'$. 
	
	Now given a set of nodes $W$ of a tree $T=(V,E)$, the set of compartments $\mathcal{C}_T(W)$ is a set of subtrees of $T$, where the union of all nodes is $V$, and for each $T' = (V',E') \in \mathcal{C}_T(W)$, $V'$ is a maximal set of nodes that induces a subtree $T''$ of $T$ such that $\delta T'' \subseteq W$.
\end{definition}

Intuitively, the set of compartments is induced by first removing $W$, so that $T$ is broken up into a forest of smaller trees, and for each of the small trees we re-add nodes in $W$ that were attached to it, where the reattached nodes are called the boundary. As opposed to partitioning, two compartments may share nodes at their boundaries. 

In the algorithm, we generate a sequence of sets $W_0 \subseteq W_1,...,W_t \subseteq V$ in the following manner: $W_0$ contains the tree median of $T$, and then to create $W_i$ from $W_{i-1}$ we simply add to $W_{i}$ the tree medians of every compartment in $\mathcal{C}_T(W_i)$.

For each $i$, we only make oracle calls of the form $f(V_{-v}(u),s)$ or $f(V_{-v}(u) \cup \{v\},s)$, and avoid choices of $(u,v)$ that will cause evaluation on overlapping sets, based on information gained on processing $W_{i-1}$ in the same way. In this way we only make essentially $O(1)$ `amortized' calls to the oracle for each $i$. For details see appendix.

After removing nodes via the reaching criterion, we only need to do this on a subtree of $T$, which we can assume takes the same time as on the full tree. One can see that $t = O(\log n)$ thus the peaking criterion takes at most $O(k \log n)$ amortized oracle calls. 

\subsubsection{Reaching criterion by Binary Search.}

Intuitively, with the reaching criterion we look for an edge $(u,v)$ in $T_{H(S)}$ so that $T_{-v}(u)$, which contains at least one sink, can be removed. 

Now given adjacent hubs $h_1$ and $h_2$, consider any subtree of $T$ rooted at $h_1$, in which $h_2$ is a descendent of $h_1$. Then exactly one of the following is true:

\begin{enumerate}[label=P\arabic{*}]
 \item There is an edge $(u,v)$ in the path $\Pi(h_1,h_2)$ between $h_1$ and $h_2$, where $u$ is a child of $v \neq h_1$, such that $T_{-v}(u)$ is recursively self-sufficient, but the subtree rooted at $v$ is not.
 \item Let $u$ be the child of $h_1$ that is on the path between $h_1$ and $h_2$. Then the subtree rooted at $u$, i.e. $T_{-h_1}(u)$, is recursively self-sufficient.
\end{enumerate}

As $h_1$ and $h_2$ are adjacent hubs, for any edge $(u,v)$ along the path, where $v \neq h_1$ is the parent of $u$, the subtree rooted at $v$ is recursively self sufficient only if the subtree rooted at $u$ is. Suppose we know that the subtree rooted at $h_2$ is recursively self-sufficient. 


In the iterative algorithm we move upwards from $h_2$ to $h_1$ gradually until we find such an edge, or upon reaching $h_1$; this can be replaced by a binary search. This idea will let us only use $O(k^2 \log n)$ calls; proper amortization with pruning can reduce this to $O(k \log n)$ oracle calls. See appendix. Theorem~\ref{theorem:FastBC} follows from  the above faster algorithm.

%% file: full_problem.tex
\section{Full problem: cost minimization}
\label{Section: Full Problem}

Given an algorithm for the bounded cost problem, it is straightforward to construct a \emph{weakly} polynomial time algorithm, by a binary search over possible values of $\mathcal{T}$ for the minimal $\mathcal{T}^*$ allowing evacuation with $k$ sinks. To produce a \emph{strongly} polynomial time algorithm, at a higher level, we wish to search among a finite, discrete set of possible values for $\mathcal{T}^*$.
This can be done by a {\em parametric searching} technique. 


\subsection{Iterative approach} 

We start by modifying the iterative algorithm for  bounded cost.  In that  algorithm, the specific value  of $\mathcal{T}$ dictates the contents of $T$, $S$, $S_\mathrm{out}$ etc., as well as which node pairs satisfy either of the two peaking criteria, at each step of Algorithms  \ref{alg:PCClimb}, \ref{alg:RCClimb}, \ref{alg:BoundedCostFull},; all these depend upon  the outcomes of comparisons of the form $f(\cdot,\cdot) \leq \mathcal{T}$.

The idea is to run a prametric search version of  Algorithm \ref{alg:BoundedCostFull}.   $\mathcal{T}$  will no longer be a constant; we {\em interfere} with the normal course of the algorithm by changing $\mathcal{T}$ during runtime. The decision to interfere is based on a \emph{threshold margin} $(\mathcal{T}^L, \mathcal{T}^H]$ that we maintain, to keep track of candidate values of $\mathcal{T}^*$. Initially, $(\mathcal{T}^L, \mathcal{T}^H] = (-\infty.+\infty]$, and $\mathcal{T} = 0$.

We step through Algorithm \ref{alg:BoundedCostFull}. Every time we evaluate $a = f( \cdot, \cdot )$, we set $\mathcal{T}$ based on the following, before making the comparison $a \leq \mathcal{T}$ and proceeding with the if-clause.

\begin{enumerate}
 \item If $a \leq \mathcal{T}^L$, set $\mathcal{T} = \mathcal{T}^L$, so the if-clause always resolves as $f(\cdot,\cdot) \leq \mathcal{T}$.
 \item If $a > \mathcal{T}^H$, set $\mathcal{T} = \mathcal{T}^H$, so the if-clause always resolves as $f(\cdot,\cdot) > \mathcal{T}$.
 \item If $a \in (\mathcal{T}^L, \mathcal{T}^H]$, run a separate clean, non-interfered instance of Algorithm \ref{alg:BoundedCostFull} with value $\mathcal{T} := a$, and observe the output.
 \begin{itemize}
  \item Output is `No': set $\mathcal{T}^L := a$, and $\mathcal{T} := a$, resolving the if-clause as $a = f(\cdot,\cdot) \leq \mathcal{T} = a$.
  \item Otherwise, set $\mathcal{T}^H := a$, and $\mathcal{T} := \mathcal{T}^L$.
 \end{itemize}
\end{enumerate}

This terminates with some $\mathcal{T} \in (\mathcal{T}^L, \mathcal{T}^H]$. We call this `Algorithm \ref{alg:BoundedCostFull} with interference'.

\begin{lemma}
 Let $(\mathcal{T}_{<},\mathcal{T}_{>}]$ be the threshold margin at the end of Algorithm \ref{alg:BoundedCostFull} with interference. Then $\mathcal{T}_{>} = \mathcal{T^*}$. In particular, we can then run Algorithm \ref{alg:BoundedCostFull} (non-interfered) on $\mathcal{T} := \mathcal{T}_{>}$ to retrieve the optimal feasible configuration.
 \label{lemma:UpperThresholdIsAnwser}
\end{lemma}
\begin{proof}[Proof of Lemma \ref{lemma:UpperThresholdIsAnwser}]
	First note that $\mathcal{T}^* \in (\mathcal{T}_{<},\mathcal{T}_{>}]$; $\mathcal{T}^* \leq \mathcal{T}_{>}$ because $\mathcal{T}^H$ is always set to be a feasible value of $\mathcal{T}$. Similarly, $\mathcal{T}^* > \mathcal{T}_{<}$ because $\mathcal{T}^L$ is always set to be a non-feasible value of $\mathcal{T}$.
	
	Then we show that, for any $\mathcal{T}_0 \in [\mathcal{T}_{<},\mathcal{T}_{>})$ (note the difference in half-openness of the interval), roughly speaking, the interfered algorithm runs in the same way as a non-interfered algorithm with $\mathcal{T} = \mathcal{T}_0$; more concretely, all if-clauses at line $\ref{alg:PCClimb:comparison}$ of $\Call{PC.Climb}$ and line $\ref{alg:RCClimb:comparison}$ of $\Call{RC.Climb}$ are resolved as if we ran the non-interfered algorithm with $\mathcal{T} = \mathcal{T}_0$.
	
	At line $\ref{alg:PCClimb:comparison}$ of $\Call{PC.Climb}$ and line $\ref{alg:RCClimb:comparison}$ of $\Call{RC.Climb}$, note that if $f(\cdot, \cdot)$ evaluates to $a < \mathcal{T}^*$, this implies $a \leq \mathcal{T}_{<} < \mathcal{T}_0$, and also the algorithm proceeds to resolve the if-clause with $f(\cdot, \cdot) \leq \mathcal{T} := a$, which is consistent with $f(\cdot, \cdot) \leq \mathcal{T}_0$. On the other hand, if $f(\cdot, \cdot)$ evaluates to $b \geq \mathcal{T}^*$, this implies $\mathcal{T}$ will be set to $\min(\mathcal{T},b)$, i.e. $b \geq \mathcal{T}_{>}$, and the algorithm proceeds to resolve the if-clause with $f(\cdot, \cdot) > \mathcal{T} := b  \geq \mathcal{T}_{>} > \mathcal{T}_0$, which is consistent with resolving with $f(\cdot, \cdot) > \mathcal{T}_0$.
	
	This, in turn, shows that Algorithm \ref{alg:BoundedCostFull} behaves in exactly the same way for any $\mathcal{T} \in [\mathcal{T}_{<},\mathcal{T}_{>})$. But we know that $\mathcal{T}_{<} < \mathcal{T}^*$, so $\mathcal{T}_{>}$ is the smallest value that is feasible i.e. at least $\mathcal{T}^*$, implying $\mathcal{T}^* = \mathcal{T}_{>}$.  \
\end{proof}

\begin{theorem}
 Minmax tree facility location can be solved in $O(n^2)$ calls to $\mathcal{A}$.
\end{theorem}

\begin{proof}
 We always allow the interfered algorithm to make progress, albeit with changing values of $\mathcal{T}$, so Lemma \ref{lemma:BoundedCostCalls} still applies; $f( \cdot, \cdot)$ is evaluated at most $O(n)$ times in the interfered algorithm, thus we also launch a separate instance of Algorithm \ref{alg:BoundedCostFull} at most $O(n)$ times. 
\end{proof}

\subsection{Using divide-and-conquer and binary search} 

The above idea still works for applying RC, that we interfere whenever we evaluate $f(\cdot,\cdot)$. Thus we only interfere $O(k \log n)$ times, making $O(k^2 \log^2 n)$ total calls to the oracle.
%

But it does not work well with the peaking criterion; that the divide-and-conquer algorithm for the peaking criterion relies very strongly on amortization, and a naive application of interference will perform $O(n)$ feasibility tests, while we aim for $O(k \log n)$.

The basic idea is to filter through values of $f(\cdot,\cdot)$ were we decide to interfere. Intuitively, the divide and conquer algorithm can be organized in $t$ layers in reference to $W_1,...,W_t$ where $t = O(\log n)$, for each we evaluate $f(\cdot, \cdot)$ on certain pairs of nodes and sets. Each evaluation of $f(\cdot, \cdot)$ can be identified with an edge of $T$, thus in each layer we have at most $O(n)$ evaluations, producing a list of $O(n)$ values.

Thus, at each layer we evaluate $f(\cdot, \cdot)$, and binary search for a pair of values $a_<,a_>$ such that $a_< \leq \mathcal{T}^* < a_>$, making $O(\log n)$ calls to the bounded-cost algorithm, and then set $\mathcal{T} = a_<$ when proceeding to mark nodes and place sinks, before moving to the next layer.

This gives $O(\log^2 n)$ calls for a single application of the peaking criterion. As we only need to apply the peaking criterion $O(k)$ times, the resulting number of calls to the feasibility test is $O(k \log^2 n)$. Theorem~\ref{theorem:kSinkRunTime} then follows.

\subsection{Evacuation time with fixed sinks (optimal partitioning)}

Here we discuss the case where we have no control over the sink placement. Given $k$ sinks at leaves, and a suitable threshold $\mathcal{T}$, applying the reaching criterion will remove all nodes from the graph. The minimum such threshold can be considered the minimum time required to evacuate all nodes with only currently placed sinks; by finding this we can supersede the tree partitioning algorithm of Mamada et al. \cite{Mamada2005}. The cost minimization algorithm will follow a similar flavor as the above, except the peaking criterion will never need to be invoked; as a result, the time complexity is $O(k^2 \log^2 n)$ oracle calls, or $O(n k^2 \log^4 n)$ time.

%% file: conclusion.tex
\section{Conclusion}
Given a Dynamic flow network on a tree $G=(V,E)$  we derive an algorithm  for finding the locations of $k$ sinks that  minimize the maximum time needed to evacuate the entire graph.  Evacuation is modelled using dynamic confluent flows. All that was previously known was an  $O(n \log^2 n)$ time algorithm for solving the one-sink ($k=1$) case.  This paper gives the first polynomial time algorithm for solving the arbitrary $k$-sink  problem.

The algorithm was developed in two parts. Section \ref{Section: Bounded Cost}  derived an $O(n k\log^3 n)$ algorithm for finding a  placement of $k$ sinks that permits evacuating the tree in 
 $\le \mathcal{T}$ time for inputted $\mathcal{T}$ (or deciding that such a placement does not exist).   Section \ref{Section: Full Problem} showed how to modify this to an $O(n k^2 \log^5 n)$ algorithm for finding the minimum such $\mathcal{T}$ that permits evacuation.

%% file: appendix.tex

\section{More details on tree climbing}

For reference we present pseudo-code for some subroutines used in the bounded cost algorithm. Algorithms~\ref{alg:CommitBlock} and \ref{alg:FindPartitionRecursiveSS} illustrate simple subroutines related to maintaining $\mathcal{P}$. Algorithms~\ref{alg:PCClimb}, $\ref{alg:RCClimb}$ and \ref{alg:BoundedCostFull} then describe the overall iterative bounded-cost algorithm.

\begin{algorithm}[h]
	\begin{algorithmic}[1]
		\State $T' = (V',E')$, rooted at $v \in V'$, sinks $S' \subseteq V'$
		\Comment $T'$ is recursively self-sufficient wrt $S'$
		\State $W : V' \rightarrow S'$, where $W(u)$ is a witness to Lemma \ref{lemma:RecursiveSS} for subtree rooted at $u$
		\State$\{P_s : s \in S'\}$, a collection of sets, all initialized to empty
		
		\State $T_0 := T'$
		\Comment We will delete nodes from $T_0$, so $T_0$ may become a forest
		
		\While{$T_0$ is non-empty}
		\State $T'_0 := $ arbitrary connected component of $T_0$, viewed as rooted subtree of $T'$
		\State $v := $ root of $T'_0$
		\State $s := W(v)$
		\State $P_s := P_s \cup \mathrm{BP}(v,s)$
		\State Remove all nodes in $\mathrm{BP}(v,s)$ from $T_0$
		\EndWhile
		
		\State $\{P_s : s \in S'\}$ is a partition witnessing self-sufficiency of $T'$
	\end{algorithmic}
	\caption{Finding partition for recursively self-sufficient trees}
	\label{alg:FindPartitionRecursiveSS}
\end{algorithm}

\begin{algorithm}[h]
	\begin{algorithmic}[1]
		\State $T = (V,E)$
		\State FIFO Queue $\mathcal{Q}_{\mathrm{PC}}$ over $V \times V$
		\State FIFO Queue $\mathcal{Q}_{\mathrm{RC}}$ (empty)
		\State $\mathtt{Marked}_{\mathrm{PC}} := \emptyset \subseteq V$
		\State $\hat{T}_{\mathrm{PC}} := T$
		
		\Procedure{PC.Climb}{}
		
		\If{$\mathcal{Q}_{\mathrm{PC}}$ is not empty}
		\State Dequeue $(u,v)$ from $\mathcal{Q}_{\mathrm{PC}}$
		\If{$f(V_{-v}(u) \cup \{v\},v) \leq \mathcal{T}$}
		\label{alg:PCClimb:comparison}
		\If{$v$ is a leaf of $\mathtt{Marked}_{\mathrm{PC}}$}
		\State $v' := $ parent of $v$ in $\hat{T}_{\mathrm{PC}}$
		\State Enqueue $(v,v')$ to $\mathcal{Q}_{\mathrm{PC}}$
		\label{alg:PCClimb:invariant}
		\State Add $v$ to $\mathtt{Marked}_{\mathrm{PC}}$
		\State Remove $v$ from $\hat{T}_{\mathrm{PC}}$
		\EndIf
		\Else
		\Comment Invoke peaking criterion
		\State Remove $T_{-v}(u)$ from $T$
		\State Add sink $v$ to $S$ and $S_{\mathrm{out}}$
		\State Commit $V_{-v}(u) \cup \{v\}$
		\EndIf
		\EndIf
		\EndProcedure
		
	\end{algorithmic}
	\caption{Tree climbing 1}
	\label{alg:PCClimb}
\end{algorithm}

%

For the peaking criterion, the iterative algorithm for the bounded cost problem repeats $\Call{PC.Climb}$ (Algorithm \ref{alg:PCClimb}) until $\mathcal{Q}_{\mathrm{PC}}$ is empty.

In the beginnining we identify the set of leaves $L$ of $T = T_{\mathrm{in}}$. Note that for any $u \in L$, $u$ supports $\{u\}$, offering a starting point for the peaking criterion. Any leaf $u \in L$ has exactly one neighbor $v$ in $T$; for every such $u \in L$, we add the pair $(u,v)$ to the FIFO queue $\mathcal{Q}_{\mathrm{PC}}$.

We also maintain a set $\mathtt{Marked}_{\mathrm{PC}}$ of nodes of $T_{\mathrm{in}}$, initially empty. A subtree $\hat{T}_{\mathrm{PC}}$ of $T$ is maintained where nodes in $\mathtt{Marked}_{\mathrm{PC}}$ are removed. Whenever we add an ordered pair of the form $(u,v)$ to $\mathcal{Q}_{\mathrm{PC}}$, $u$ is marked i.e. put into $\mathtt{Marked}_{\mathrm{PC}}$, and removed from $\hat{T}_{\mathrm{PC}}$.

By the end of this process, $T$ may still contain some nodes, but it is guaranteed to be RC-viable. Then, we can start applying the reaching criterion. $\hat{T}_{\mathrm{PC}}$ is the tree induced by $V \backslash \mathtt{Marked}_{\mathrm{PC}}$, and at this point is incidentally the hubtree $T_{H(S)}$.

Similar to the above, we have an other FIFO queue $\mathcal{Q}_{\mathrm{RC}}$ that contains ordered node pairs, and a set of nodes  $\mathtt{Marked}_{\mathrm{RC}}$ of $\hat{T}_{\mathrm{PC}}$, initialized to be empty, with a corresponding tree $T_{\mathrm{RC}}$. For technical reasons whenever a node is put in $\mathtt{Marked}_{\mathrm{RC}}$, it is also  put in $\mathtt{Marked}_{\mathrm{PC}}$. Initially, for every sink $s$ in $S$, which is now a leaf of $\hat{T}_{\mathrm{PC}}$, we take its parent $t$ in $\hat{T}_{\mathrm{PC}}$ and enqueue $(s,t)$ to $\mathcal{Q}_{\mathrm{RC}}$. 

\begin{algorithm}[ht]
	\begin{algorithmic}[1]
		\State $T = (V,E)$
		\State FIFO Queue $\mathcal{Q}_{\mathrm{PC}}$
		\State $\mathtt{Marked}_{\mathrm{PC}} \subseteq V$
		\State $\mathtt{Marked}_{\mathrm{RC}} := \empty \subseteq V \backslash
		\mathtt{Marked}_{\mathrm{PC}}$
		\State $\hat{T}_{\mathrm{RC}}$
		\State $\hat{T}_{\mathrm{PC}}$
		
		\Procedure{RC.Climb}{}
		
		\If{$\mathcal{Q}_{\mathrm{PC}}$ is empty}
		\If{$\mathcal{Q}_{\mathrm{RC}}$ is not empty}
		\State Dequeue $(u,v)$ from $\mathcal{Q}_{\mathrm{RC}}$
		\State $S' := $ set of sinks in $T_{-v}(u)$
		\For{Each $s \in S'$}
		\If{$f(\mathrm{BP}(v,s), s) \leq \mathcal{T}$}
		\label{alg:RCClimb:comparison}
		\State Put $v$ in $\mathtt{Marked}_{\mathrm{RC}}$ and $\mathtt{Marked}_{\mathrm{PC}}$
		\If{$v$ is leaf of $\hat{T}_{\mathrm{RC}}$}
		\State $v' := $ parent of $v$ in $\hat{T}_{\mathrm{RC}}$
		\State Enqueue $(v,v')$ to $\mathcal{Q}_{\mathrm{RC}}$
		\EndIf
		\State Remove $v$ from $\hat{T}_{\mathrm{RC}}$,$\hat{T}_{\mathrm{PC}}$
		\State Break for loop
		\Else
		\Comment Invoke reaching criterion
		\State Commit blocks for $T_{-v}(u)$
		\State Remove $T_{-v}(u)$ from $T$, $\hat{T}_{\mathrm{PC}}$ and $\hat{T}_{\mathrm{RC}}$
		\If{$v$ is leaf of $\hat{T}_{\mathrm{PC}}$}
		\State $v' := $ parent of $v$ in $\hat{T}_{\mathrm{PC}}$
		\State Enqueue $(v,v')$ to $\mathcal{Q}_{\mathrm{PC}}$
		\label{alg:RCClimb:invariant1}
		\EndIf
		\EndIf
		\EndFor
		\EndIf
		\EndIf
		\EndProcedure
		
	\end{algorithmic}
	\caption{Tree climbing 2}
	\label{alg:RCClimb}
\end{algorithm}

Then whenever $\mathcal{Q}_{\mathrm{PC}}$ is empty but $\mathcal{Q}_{\mathrm{RC}}$ is not, we call $\Call{RC.Climb}$ (Algorithm \ref{alg:RCClimb}) which carries out a test for the reaching criterion.

\begin{algorithm}[ht]
	\begin{algorithmic}[1]
		\State Given $\mathcal{T}$, $T_\mathrm{in} = (V_{\in},E_{\in})$
		\State $\hat{T}_{\mathrm{RC}} := \hat{T}_{\mathrm{PC}} := T := (V,E) := T_{\mathrm{in}}$
		\State $S := S_\mathrm{out} = \emptyset, \mathcal{P}_{\mathrm{out}} := \emptyset$. 
		
		\State $\mathcal{Q}_{\mathrm{PC}}$, $\mathcal{Q}_{\mathrm{RC}}$ empty
		
		
		\For{Each leaf $u$ of $T$}
		\State $v := $neighbor of $u$ in $\hat{T}_{\mathrm{PC}}$
		\State Enqueue $(u,v)$ to $\mathcal{Q}_{\mathrm{PC}}$
		\State Remove $u$ from $\hat{T}_{\mathrm{PC}}$ , $\hat{T}_{\mathrm{RC}}$
		\EndFor
		\Repeat
		\While{$\mathcal{Q}_{\mathrm{PC}}$ is not empty}
		\Call{PC.Climb}{}
		\EndWhile
		\Call{RC.Climb}{}
		\Until{$\mathcal{Q}_{\mathrm{RC}}$, $\mathcal{Q}_{\mathrm{PC}}$ are both empty}
		\State Output $S_\mathrm{out}$, $\mathcal{P}_{\mathrm{out}}$
	\end{algorithmic}
	\caption{Bounded cost algorithm}
	\label{alg:BoundedCostFull}
\end{algorithm}

For the correctness of Algorithm \ref{alg:BoundedCostFull}, it suffices to show that we maintain these invariants in Algorithms \ref{alg:PCClimb} and \ref{alg:RCClimb}:

\begin{enumerate}[label=\textrm{IVQ}\arabic{*}]
	\item Whenever we enqueue $(u,v)$ in $\mathcal{Q}_{\mathrm{PC}}$, we know that $f(V_{-v}(u) ,u) \leq \mathcal{T}$.
	\item $(v,v')$ is enqueued to $\mathcal{Q}_{\mathrm{RC}}$ at some point in the algorithm if and only if $T_{-v'}(v)$ is recursively self-sufficient.
\end{enumerate}

\begin{lemma}
	Invariants $\mathrm{IVQ1}$ and $\mathrm{IVQ2}$ are maintained for both Algorithms \ref{alg:PCClimb} and Algorithm \ref{alg:RCClimb}, thus also throughout Algorithm \ref{alg:BoundedCostFull}.
\end{lemma}
\begin{proof}	
	To see this for Algorithm \ref{alg:PCClimb}, note that only line \ref{alg:PCClimb:invariant} enqueues $(v,v')$ to $\mathcal{Q}_{\mathrm{PC}}$, which happens only when for every neighbor $u$ of $v$ (except $v'$) $u$ is marked, i.e. $f(V_{-v}(u) \cup \{v\},v) \leq \mathcal{T}$, which by max composition implies $f(V_{-v'}(v),v) \leq \mathcal{T}$. Thus Algorithm \ref{alg:PCClimb} applies the peaking criterion correctly.
	
	To see this is for Algorithm \ref{alg:RCClimb}, first note that by only proceeding when $\mathcal{Q}_{\mathrm{PC}}$ is empty, we ensured that $T$ is RC-viable. We then show the invariants inductively. We are given $v$ and a neighbor $v'$.
	
	First of all it is trivial if $v$ is a sink. Now suppose by the inductive hypothesis, every neighbor $u$ of $v$ in $T$ where $T_{-v}(u)$ has a sink, except $v'$, is such that $T_{-v}(u)$ is recursively self-sufficient. By the processing order and IH, we know that $(u,v)$ must have been enqueued previously and $u$ is in $\mathtt{Marked}_{\mathrm{RC}}$, thus we will arrive at the pair $(v,v')$. Then $(v,v')$ is enqueued only if for some sink $s$ in $T_{-v'}(v)$ we have $f(\mathrm{BP}(v',s),s) \leq \mathcal{T}$, which by Lemma \ref{lemma:RecursiveSS} and IH means $T_{-v'}(v)$ is recursively self-sufficient.
	
	Conversely, if $T_{-v'}(v)$ is not recursively self-sufficient, then either for one neighbor $u$ of $v$ except $v'$ we have that $T_{-v}(u)$ is not recursively self-sufficient, or that there exists $u$ such that for every sink $s$ in $T_{-v}(u)$ we have $f(\mathrm{BP}(v',s),s) > \mathcal{T}$. In the former case, $(u,v)$ would not have entered the queue yet, and in the latter case we will remove $T_{-v}(u)$, thus in either case we do not enqueue $(v,v')$.
	
	Note that line \ref{alg:RCClimb:invariant1} also enqueues to $\mathcal{Q}$, thus we need that invariant $\mathrm{IVQ1}$ is also maintained. Note that $T$ is known to be RC-viable when \Call{PC.Climb}{} is called, thus $v$ can serve  outstanding branches attached to it. On the other hand, for all neighbors of $u$ except $v'$ where $T_{-v}(u)$ contains a sink, $T_{-v}(u)$ is not in $\hat{T}_{\mathrm{PC}}$ (removed due to reaching criterion), so $T_{-v'}(v)$ only contains $v$ and outstanding branches attached to it, thus $T_{-v'}(v) \leq \mathcal{T}$ by max-composition.
\end{proof}

\section{Detailed description of recursive algorithm for PC}
We recall the definition of compartments.

\begin{definition}[Compartments and Boundaries]
	Let $T'=(V',E')$ be a subtree of $T=(V,E)$. The boundary $\delta T'$ of $T'$ is defined to be the set of every nodes in $T'$ that is a neighbor of some node in $V \backslash V'$.
	
	Now given a set of nodes $W$ of a tree $T=(V,E)$, the set of compartments $\mathcal{C}_T(W)$ is a set of subtrees of $T$ with the following properties:
	
	\begin{enumerate}[label=\arabic{*}.]
		\item For each $T' = (V',E') \in \mathcal{C}_T(W)$, $V'$ is a maximal set of nodes that induces a subtree $T''$ of $T$ such that $\delta T'' \subseteq W$.
		\item $\cup_{(V',E') \in \mathcal{C}_T(W)} V' = V$
	\end{enumerate}
\end{definition}

Intuitively, the set of compartments is induced by first removing $W$, so that $T$ is broken up into a forest of smaller trees, and for each of the small trees we re-add nodes in $W$ that were attached to it, where the reattached nodes are called the boundary. As opposed to partitioning, two compartments may share nodes at their boundaries. 

Then given $W$ and $\mathtt{Marked}_\mathrm{PC}$, we define the process of `calling the oracle on $\mathcal{C}_T(W)$', which does the following for every $T' = (V',E') \in \mathcal{C}_T(W)$.
If $w$ is the only node in $\delta T'$ that is not marked, we look at the subtree of the form $T_{-w}(u) \cup \{w\}$ where $u \in V'$. If $T_{-w}(u) \cup \{w\}$ does not contain a sink, then we treat $(u,w)$ as a possible pair for PC, and evaluate  $f(V_{-w}(u) \cup \{w\},w)$. Then we act according to the outcome:
\begin{itemize}
\item If $f(V_{-w}(u) \cup \{w\},w) \leq \mathcal{T}$ then the entire $V_{-w}(u) \cup \{w\}$ is marked.
\item If $f(V_{-w}(u) \cup \{w\},w) > \mathcal{T}$, evaluate $f(V_{-w}(u),u)$; if now $f(V_{-w}(u),u) \leq \mathcal{T}$, the peaking criterion is invoked, a sink is placed at $w$, and all nodes in $U$ would be marked.
\end{itemize}

We can see that the oracle is specifically not called in the following circumstances:

\begin{enumerate}[label=\textrm{R}\arabic{*}]
	\item If more than one node in $\delta T'$ is not marked. \label{bulletp:Recursion1}
	\item If all nodes in $\delta T'$ are already marked.\label{bulletp:Recursion2}
\end{enumerate}

%
%
%

From this, we will see soon that when we call the oracle on $\mathcal{C}_T(W)$, we apply the oracle to sets totalling $O(n)$ nodes, or exactly $n-1$ edges. Thus the total time is equivalent to the time running the oracle on $T$, with constant factor (and additive linear) overhead. For this reason we also call this an amortized oracle call.

To make sure this is the case, we need to specify how $W$ is constructed and $\mathtt{Marked}_\mathrm{PC}$ is formed as we progress.

Let $\rho$ be the tree median of $T$, and let $W_0 := \{ \rho \}$. The  set of compartments $\mathcal{C}_T(W_0)$ then simply consists of trees of the form $T_{-\rho}(u) \cup \{\rho\}$ for every $u \in N(\rho)$. We make an oracle call on $\mathcal{C}_T(W_0)$, and mark nodes as given above.

Then we create a new set $W_1$, first a copy of $W_0$, and for every compartment $T'$ of $\mathcal{C}_T(W_0)$, we take the tree median $\rho'$ of $T' \backslash \delta T'$ and put it in $W_1$.

We again make an oracle call on $\mathcal{C}_T(W_1)$, mark nodes and place sinks when applicable, and create $W_2$ in a similar fashion, and so on.

This generates a sequence of subsets $W_0$,..,$W_t$ for some $t \geq 0$, where we stop when $W_t = V$. The period from creating $W_i$ to making an oracle call on $\mathcal{C}_T(W_i)$ is called epoch $i$.

Note that $t$ can be at most $O(\log n)$, so we claim that this makes $O(\log n)$ amortized oracle calls. Both the time bound and the correctness need to be established through the following observation. 


\begin{lemma}
	Given a subtree of the form $T_{-v}(u)$, suppose $f(V_{-v}(u) \cup \{v\}, v) \leq \mathcal{T}$. Let $r \geq 0$ be the smallest number such that $v \in W_r$. Then exactly one of the following will occur:
	
	\begin{enumerate}[label=Case \arabic{*}]
		\item $v$ was already marked during an earlier epoch $i < r$, 
		\item $T_{-v}(u)$ contains a sink, placed in an earlier epoch $j < r$.

		\item The algorithm never makes an oracle call of the form $f(V_{-v'}(u') \cup \{v'\}, v')$ where $(V_{-v}(u) \subseteq V_{-v'}(u') \cup \{v'\}$ in any epoch, and $f(V_{-v}(u), v) > \mathcal{T}$
		\item The algorithm evaluates $f(V_{-v}(u) \cup \{v\}, v)$ and $f(V_{-v}(u),u)$ in epoch $r$.
	\end{enumerate}
	\label{lemma:ExclusionRecusrivePC}
\end{lemma}

Before we proceed to the proof, we note that the second case rules out the possibility that a sink placed in epoch $r$ will affect the process of making a oracle call on $\mathcal{C}_T(W_r)$, so that the oracle call on $\mathcal{C}_T(W_r)$ is well-defined.

The rationale of the third case is that the algorithm already knows implicitly $f(V_{-v}(u), v) > \mathcal{T}$ hence also knows that no sink outside $V_{-v}(u)$ can serve a set that contains $V_{-v}(u)$. Even better, this information also propagates between compartments; that a compartment $\mathcal{C}$ far from $u,v$ `knows' that no sink placed within itself can serve a subtree of $T$ rooted in $\mathcal{C}$ that contains $V_{-v}(u)$, just by observing whether some of its boundary nodes are marked, hence we can avoid calling the oracle excessively.

\begin{proof}

	We prove by induction on $r$. If $r = 0$ then the algorithm must evaluate $f(V_{-v}(u) \cup \{v\}, v)$.
	
	Now assume IH to be true for $r = r' \geq 0$. Consider $r = r' + 1$.
	
	The mutual exclusion among cases 1,2 and 4 is clear, so suppose the algorithm does not evaluate $f(V_{-v}(u) \cup \{v\}, v)$ at epoch $r'+1$, and $v$ was not already marked in any earlier epoch i.e. not at the start of epoch $r'+1$, and also assume no sink is placed in $V_{-v}(u)$.
	
	Let $T' = (V',E')$ be the compartment $\mathcal{C} \in \mathcal{C}_T(W_{r'+1})$ that contains $u$ (as well as $v$). Then there exists another node $v'$ on the boundary of $T'$ (i.e. $v' \in \delta T'$) such that $v'$ is not marked at the start of epoch $r'+1$. 
	
	Let $u'$  be any neighbor of $v'$ \emph{outside} $T'$.  
	
	Note that $V_{-v'}(u') \subseteq V_{-v}(u)$, so by assumption $V_{-v'}(u')$ must also contain no sink. Because $v'$ is also not marked, IH implies that either the algorithm evaluated $f(V_{-v'(u')}\cup \{v'\}, v')$ or it was already known in epoch $r'$ that $f(V_{-v'}(u'),u') > \mathcal{T}$. The latter case settles the proof. In the former case, we again have two possibilities: either $f(V_{-v'(u')}\cup \{v'\}, v') \leq \mathcal{T}$ or $f(V_{-v'(u')}\cup \{v'\}, v') > \mathcal{T}$; in the latter case again we are done, so the remaining case is $f(V_{-v'(u')}\cup \{v'\}, v') \leq \mathcal{T}$.
	
	Note that the choice of $u'$ was arbitrary; if we can find an other neighbor $u''$ of $v'$ that is not in $T'$ so that the case $f(V_{-v'(u'')}\cup \{v'\}, v') \leq \mathcal{T}$ does not occur then we are done. Thus the actual remaining case is that $f(V_{-v'(u';)}\cup \{v'\}, v') \leq \mathcal{T}$ for \emph{every} neighbor $u''$ of $v'$ outside $T'$. However, in this case the algorithm would have marked $v'$, violating our assumptions. So IH holds for $r=r'+1$.
\end{proof}

Cases 1-3 in the above lemma characterize the circumstances where we can avoid calling the oracle. We can then prove the following.

\begin{lemma}
For any $0 \leq i \leq t$, making an oracle call on $\mathcal{C}_T(W_i)$ takes time $O(n +t_\mathcal{A}(n))$.
\end{lemma}
\begin{proof}
	Given any epoch $i$, it suffices to show that for any distinct $w_1, w_2 \in W_i$ and subtrees of the form $T_{-w_1}(u_1)$, $T_{-w_2}(u_1)$, if $V_{-w_1}(u_1) \cup \{w_1\} \subseteq V_{-w_2}(u_2) \cup \{w_2\}$. then the oracle will never evaluate both of $f(V_{-w_1}(u_1) \cup \{w_1\}, w_1)$  and $f(V_{-w_2}(u_2) \cup \{w_2\}, w_2)$. 
	
	WLOG assume $V_{-w_1}(u_1) \cup \{w_1\} \subseteq V_{-w_2}(u_2) \cup \{w_2\}$, and neither of $w_1$ and $w_2$ are marked.
	
	First suppose that the algorithm evaluates $f(V_{-w_1}(u_1) \cup \{w_1\},w_1)$. Let $\tau$ be a node on the path from $w_1$ to $w_2$, 
	where $\tau \in W_i$ but $\tau \neq w_1, w_2$. If no such node exists, then in fact there is a compartment $\mathcal{C} \in \mathcal{C}_T(W_i)$ that contains both $w_1$ and $w_2$, $u_2$, in which case we do not evaluate $f(V_{-w_2}(u_2) \cup \{w_2\}, w_2)$, because two nodes on the boundary of $\mathcal{C}$ are unmarked. So assume that $\tau$ exists. We can further assume that $\tau$ is unmarked; because otherwise either $w_1$ or $w_2$ is marked, violating our assumptions. This means we can assume that any node on the path between $w_1$ and $w_2$ within $W$ is not marked. In particular, there exists a $\tau$ in a same compartment $\mathcal{C}$ as $w_2$ that is not marked, which again forces us not to evaluate $f(V_{-w_2}(u_2) \cup \{w_2\}, w_2)$.
	
	Then suppose that the algorithm evaluates $f(V_{-w_2}(u_2) \cup \{w_2\},w_2)$. Let $\mathcal{C}$ be the compartment that contains both $w_2$ and $u_2$. This means that any node except $w_2$ on the boundary of $\mathcal{C}$ is marked; so if $w_1$ is also in $\mathcal{C}$ our assumption will be contradicted. Thus we can again assume that $w_1$ is not in $\mathcal{C}$. Then there exists $\tau \in W_i$ on the boundary of $\mathcal{C}$ along the path from $w_1$ to $w_2$ that is marked, which implies either $w_1$ is marked or $w_2$ is marked, a contradiction. Thus under our assumptions the algorithm actually never evaluates $f(V_{-w_2}(u_2) \cup \{w_2\},w_2)$.
	
	This in turn implies that in each epoch, there is a set of edge-disjoint set of subtrees $\mathcal{G}$ such that on each $T'=(V',E') \in \mathcal{G}$ we only run the oracle at most twice. Then because $t_{\mathcal{A}}(n) = \Omega(n)$, we see that the time spent on the oracle is bounded above by $2\sum_{(V',E') \in \mathcal{G}} t_{\mathcal{A}}(|E'|) \leq t_{\mathcal{A}}(n)$. The overhead is $O(n)$ so the final time bound is $O(n + t_{\mathcal{A}}(n))$.
%
%
%
%
\end{proof}

{\bf Correctness.} Consider a subtree of $T$ the form $T_{-v}(u) \cup \{v\}$ where $f(V_{-v}(u), u) \leq \mathcal{T}$ but $f(V_{-v}(u) \cup \{v\}, v) > \mathcal{T}$, i.e. the pair $u,v$ satisfies PC. We also assume non-degeneracy, that there exists no $s \in V$ where $f(V,s) \leq \mathcal{T}$, so that a sink must be placed at $u$ but not within $V_{-v}(u)$. We show that the algorithm evaluates both $f(V_{-v}(u), u)$ and $f(V_{-v}(u) \cup \{v\}, v)$, which puts a sink at $u$.

Let $r \geq 0$ be the smallest number such that $v \in W_{r}$, and consider epoch $r$. Because $f(V_{-v}(u) \cup \{v\}, v) > \mathcal{T}$, at least one node in $V_{-v}(u)$ is not marked, which in turn implies $u$ is not marked.

Now let $T' = (V',E') \in \mathcal{C}_T(W_r)$ be the compartment that contains $v$ and $u$. Pick an arbitrary $w \in \delta T'_l \backslash \{u\}$. First of all, note that by construction $v$ is the tree median of a compartment in epoch $r-1$ where $u'$ is on the boundary; so $w \in W_{r-1}$ and first appears in some $W_l$ where $l < r$.

Let $u'$ be a neighbor of $w$ that is not in $T'$. Because $f(V_{-v}(u), u) \leq \mathcal{T}$ and also due to non-degeneracy, $V_{-w}(u')$ can not contain a sink in epoch $l$. This implies either Case 1 or Case 4 in Lemma~\ref{lemma:ExclusionRecusrivePC}. In Case 1 $w$ is marked; if Case 1 does not hold, for Case 4 note that the choice of $u'$ was arbitrary, so it must hold for any neighbor $u'$ of $w$ that is not in $T'$, in which case $w$ would still have been marked after epoch $l$.

Thus we know that $w$ is already marked in epoch $r$. The choice of $w$ was again arbitrary, thus all nodes in $\delta T'_l \backslash \{u\}$ are marked, and the algorithm will indeed evaluate $f(V_{-v}(u) \cup \{v\}, v)$ and find that it exceeds $\mathcal{T}$; it will also evaluate $f(V_{-v}(u), v)$ and find that it is no more than $\mathcal{T}$. This suffices for correctness.

\section{More Details for Reaching Criterion by Binary Search}
Recall that we work with $T=(V,E)$ and a set of sinks $S$ placed at leaves of $T$ so that $T$ is RC-viable. More concretely, suppose the subtree rooted at $h_2$ is already known to be recursively self-sufficient. If $h_1$ is a neighbor of $h_2$ (thus  $|\Pi(h_1,h_2)| = 2$), we can simply do the same tests as the iterative algorithm due to  Lemma~\ref{lemma:RecursiveSS}. 

Now WLOG suppose $|\Pi(h_1,h_2)| > 2$. We assume $T$ is rooted at $h_1$. Let $v$ be the parent of $h_2$, and suppose that we already know the subtree $T_{-v}(h_2)$ is recursively self-sufficient. Lemma~\ref{lemma:RecursiveSS} tells us that the subtree rooted at $v$ is recursively self sufficient if and only if there is a sink $s$ in $T_{-v}(h_2) \cap S$ such that $v$ can evacuate to $s$, which requires at most $k$ calls to the oracle to test.

In fact, this is also the same for the parent $v$ (which is not $h_1$). Thus for any node $v' \in \Pi(h_1,h_2) \backslash \{h_1,h_2\}$ we can use the same test to test for self-sufficiency: if the tree rooted at $v'$ is found to be self-sufficient, then the tree rooted at any of its descendent is recursively self-sufficient.

This with a binary search along the path, we can find the highest node $v$ in $\Pi(h_1,h_2) \backslash \{h_1\}$ such that the subtree rooted at $v$ is recursively self-sufficient, making $O(k \log n)$ calls to the oracle. 

To count the total number of binary searches we need, note that after completing the binary search, either at least one sink is removed from $T$, or we gain new knowledge that for a hub $h$, a subtree rooted at $h$ is recursively self-sufficient. These events can only occur at most $O(k)$ times, so this way we only need $O(k^2 \log n)$ calls to the oracle.

We can adjust this process to achieve $O(k \log n)$ calls. To see this consider the following.

Let $u$ be the child of $h_1$ in $\Pi(h_1,h_2)$. Before applying a binary search along the path $\Pi(h_1,h_2)$ as above, we test whether $T_{-h_1}(u)$ is recursively self sufficient (rooted at $u$), where we only need to find a sink $s \in T_{-h_1}(u) \cap S$ such that $f(\mathrm{BP}(u,s),s) \leq \mathcal{T}$.

We analyze two outcomes:

\begin{enumerate}
	\item \emph{If $T_{-h_1}(u)$ is recursively self-sufficient}, we can mark the entire tree $T_{-h_1}(u)$.
	\item  \emph{If $T_{-h_1}(u)$ is not recursively self-sufficient}, we apply the binary search, making $O( |T_{-h_1}(u) \cap S| \log n)$ calls to the oracle to find the cut-off edge,  invoking the reaching criterion. This removes all sinks in $T_{-h_1}(u) \cap S$.
\end{enumerate}

In the first case, before proceeding with the rest of the algorithm, suppose the algorithm tests $f(\mathrm{BP}(u,s_i),s_i) \leq \mathcal{T}$ for a sequence of sinks $s_1,...,s_m$ from $T_{-h_1}(u) \cap S$, where $m \leq |T_{-h_1}(u) \cap S|$,  $f(\mathrm{BP}(u,s_m),s_m) \leq \mathcal{T}$ and $f(\mathrm{BP}(u,s_i),s_i) > \mathcal{T}$ for $i<m$.

By path monotonicity, for any node $v \notin V_{-h_1}(u)$, we would then know that $f(\mathrm{BP}(v,s_i),s_i) > \mathcal{T}$ for all $s_m$, thus there is no need to evaluate $f(\mathrm{BP}(v,s_i),s_i)$ for the remainder of the algorithm, for $i < m$. We say that sinks $s_1,...,s_{m-1}$ are \emph{rejected}, at which point no more oracle calls need to be wasted on them. We also say that $u$ accepts the sink $s_m$.

A sink can only be rejected once. On the other hand, a node in place of $u$ can only accept some sink once. In each of these two events the oracle is called exactly once. All of these events combined can only occur $O(k)$ times; the former because there can be at most $k$ sinks, the latter because there can be at most $O(k)$ hubs. Thus the total number of oracle calls made for the first case in the entire course of our algorithm is $O(k)$.

For the latter case, we can charge at most $O(k' \log n)$ calls each time we remove $k'$ sinks from $T$. As there can only be $k$ sinks to be removed in total, this costs $O(k \log n)$ calls.

Now, the last type of oracle calls for the reaching criterion are made when we check whether a tree rooted at a hub $h$ is recursively self-sufficient. In the iterative algorithm, we only check this for a hub $h$ if at most one of its neighbors in $T_{H(S)}$ is unmarked, and it is the same for this non-iterative version.

$h$ is either the only node in $T_{H(S)}$ that is unmarked, or has a natural `parent' in $T_{H(S)}$, which is its only non-marked neighbor $v'$ in $T_{H(S)}$. Thus we need to test if $T_{-v'}(v)$ is recursively self-sufficient. This is basically the same as the first case above; we reject a sink $s \in T_{-v'}(v) \cap S$ if $BP(\mathrm{BP}(v,s) > \mathcal{T}$, and when $v$ accepts a sink $s$ we declare $T_{-v'}(v)$ to be recursively self-sufficient. By the same counting argument, the number of oracle calls made for this case is also $O(k)$.

Thus overall we only need a total of $O(k \log n)$ oracle calls to test for and apply the reaching criterion throughout the algorithm.

%

\section{Omitted Proofs and Lemmas}

\section{Omitted Figures}

\begin{figure}[h]
 \centering
   \includegraphics[width=3in]{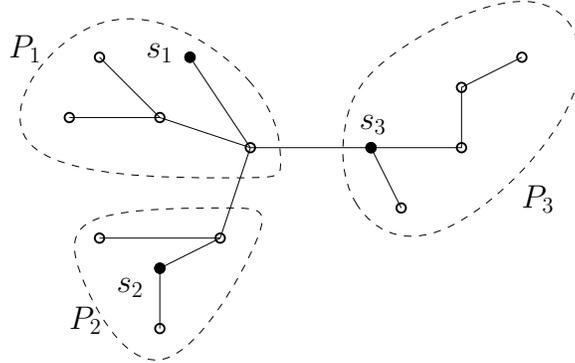}
   \label{fig:maxcomposition}
   \caption{Example of $F_S$: Partition $\mathcal{P}=\{P_1,P_2,P_3\},$, sinks $S = \{s_1,s_2,s_3\}$, $F_S(\mathcal{P}) = \max(f(P_1,s_1),f(P_2,s_2),f(P_3,s_3))$} 
\end{figure}

\begin{figure}[h]
 \centering
\includegraphics[width=0.5\textwidth]{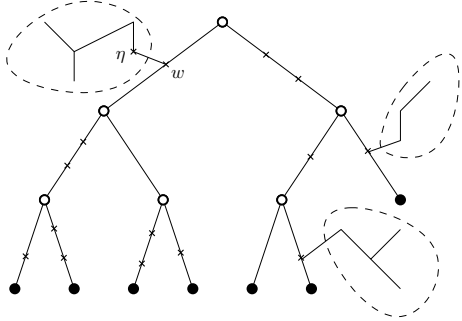}
   \label{fig:HubTree}
   \caption{Hub tree. Circled with dashed lines are outstanding branches, dark circles are sinks, and hollow circles are hubs.}
\end{figure}

\begin{figure}[h]
 \centering
\includegraphics[width=0.5\textwidth]{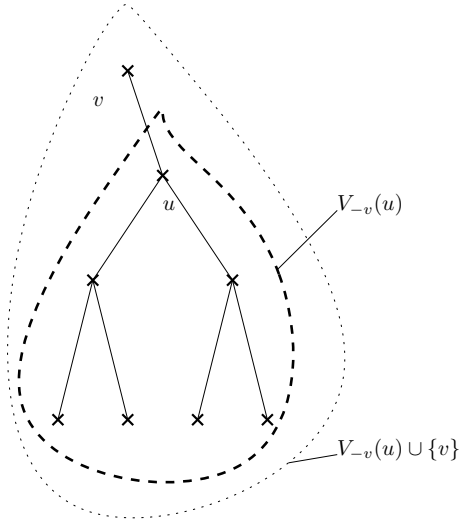}
   \label{fig:Peaking}
   \caption{Peaking criterion. Note the tree has no sinks. If $f(V_{-v}(u),u) \leq \mathcal{T}$, then $u$ can serve $f(V_{-v}(u),u)$, so no sink has to be placed below $u$; on the other hand, if $f(V_{-v}(u) \cup \{ v \},v) > \mathcal{T}$, then no node outside this figure can support $f(V_{-v}(u),u)$ single-handedly. This pinpoints the position of exactly one sink to be placed at $u$.}
\end{figure}

\begin{figure}[h]
 \centering
\includegraphics[width=0.5\textwidth]{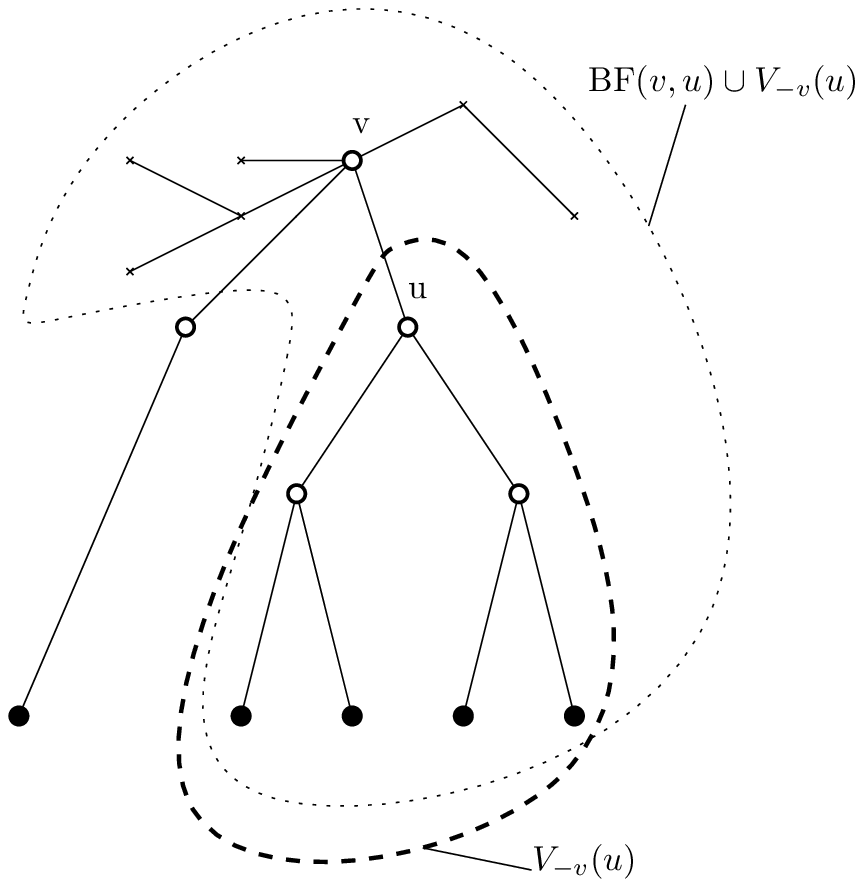}
   \label{fig:DualPeaking}
   \caption{Reaching criterion. Assume this tree is RC-viable. Dark circles are sinks, white circles are hubs. If $T_{-v}(u)$ is self-sufficient, then no extra sinks have to be put in $T_{-v}(u)$. If $BF(v,u) \cup V_{-v}(u)$ is not self-sufficient, RC-viability implies that we never have to assign $v$ to any sink $s$ in $T_{-v}(u)$; this is because assigning $v$ to a sink $s$ downwards will force all nodes in all outstanding branches attached to $v$ to also be assigned to $s$, which is not feasible unless we place a new sink in at least one of said outstanding branches, and this new sink may as well be placed at $v$, because $v$ can serve any outstanding branch attached to it (due to RC-viability), reverting the need to assign $v$ to $s$.}
\end{figure}

%% file: arxiv.bbl
\begin{thebibliography}{10}

\bibitem{Aronson1989}
J.~E. Aronson.
\newblock {A survey of dynamic network flows}.
\newblock {\em Annals of Operations Research}, 20(1):1--66, 1989.
\newblock URL: \url{http://link.springer.com/article/10.1007/BF02216922}.

\bibitem{ArumugamAGS15}
G.P. Arumugam, J.~Augustine, M.~J. Golin, and P.Srikanthan.
\newblock Optimal evacuation on dynamic paths with general capacities of edges.
\newblock {\em Unpublished Manuscript}, 2015.

\bibitem{Chen2007}
Jiangzhuo Chen, Robert~D Kleinberg, L{\'a}szl{\'o} Lov{\'a}sz, Rajmohan
  Rajaraman, Ravi Sundaram, and Adrian Vetta.
\newblock {(Almost) Tight bounds and existence theorems for single-commodity
  confluent flows}.
\newblock {\em Journal of the ACM}, 54(4), jul 2007.

\bibitem{Chen2006}
Jiangzhuo Chen, Rajmohan Rajaraman, and Ravi Sundaram.
\newblock {Meet and merge: Approximation algorithms for confluent flows}.
\newblock {\em Journal of Computer and System Sciences}, 72(3):468--489, 2006.

\bibitem{Dressler2010b}
Daniel Dressler and Martin Strehler.
\newblock {Capacitated Confluent Flows: Complexity and Algorithms}.
\newblock In {\em 7th International Conference on Algorithms and Complexity
  (CIAC'10)}, pages 347--358, 2010.

\bibitem{Fleischer2007}
Lisa Fleischer and Martin Skutella.
\newblock {Quickest Flows Over Time}.
\newblock {\em SIAM Journal on Computing}, 36(6):1600--1630, January 2007.
\newblock URL: \url{http://epubs.siam.org/doi/abs/10.1137/S0097539703427215},
  \href {http://dx.doi.org/10.1137/S0097539703427215}
  {\path{doi:10.1137/S0097539703427215}}.

\bibitem{Ford1958a}
L.~R. Ford and D.~R. Fulkerson.
\newblock {Constructing Maximal Dynamic Flows from Static Flows}.
\newblock {\em Operations Research}, 6(3):419--433, June 1958.

\bibitem{frederickson1991parametric}
Greg~N Frederickson.
\newblock Parametric search and locating supply centers in trees.
\newblock In {\em Proceedings of the Second Workshop on Algorithms and Data
  Structures (WADS'91)}, pages 299--319. Springer, 1991.

\bibitem{garey1979computers}
Michael~R Garey and David~S Johnson.
\newblock {\em Computers and intractability: A Guide to the Theory of
  NP-Completeness}.
\newblock W.H. Freeman and Company, 1979.

\bibitem{Higashikawa2014}
Y.~Higashikawa, M.~J. Golin, and N.~Katoh.
\newblock {Minimax Regret Sink Location Problem in Dynamic Tree Networks with
  Uniform Capacity}.
\newblock In {\em Proc of the 8'th Intl Workshop on Algorithms and Computation
  (WALCOM'2014)}, pages 125--137, 2014.

\bibitem{higashikawa2015multiple}
Yuya Higashikawa, Mordecai~J Golin, and Naoki Katoh.
\newblock Multiple sink location problems in dynamic path networks.
\newblock {\em Theoretical Computer Science}, 607:2--15, 2015.

\bibitem{Hoppe2000b}
B~Hoppe and \'{E} Tardos.
\newblock {The quickest transshipment problem}.
\newblock {\em Mathematics of Operations Research}, 25(1):36--62, 2000.

\bibitem{Kamiyama}
Naoyuki Kamiyama, Naoki Katoh, and Atsushi Takizawa.
\newblock {Theoretical and Practical Issues of Evacuation Planning in Urban
  Areas}.
\newblock In {\em The Eighth Hellenic European Research on Computer Mathematics
  and its Applications Conference (HERCMA2007)}, pages 49--50, 2007.

\bibitem{Mamada2005a}
Satoko Mamada and Kazuhisa Makino.
\newblock {An Evacuation Problem in Tree Dynamic Networks with Multiple Exits}.
\newblock In {Tatsuo Arai}, Shigeru Yamamoto, and Kazuhi Makino, editors, {\em
  Systems \& Human Science-For Safety, Security, and Dependability; Selected
  Papers of the 1st International Symposium SSR2003}, pages 517--526. Elsevier
  B.V, 2005.

\bibitem{Mamada2005}
Satoko Mamada, Takeaki Uno, Kazuhisa Makino, and Satoru Fujishige.
\newblock {A tree partitioning problem arising from an evacuation problem in
  tree dynamic networks}.
\newblock {\em Journal of the Operations Research Society of Japan},
  48(3):196--206, 2005.

\bibitem{Mamada2006}
Satoko Mamada, Takeaki Uno, Kazuhisa Makino, and Satoru Fujishige.
\newblock {An $O(n \log^2 n) $algorithm for the optimal sink location problem
  in dynamic tree networks}.
\newblock {\em Discrete Applied Mathematics}, 154(2387-2401):251--264, 2006.

\bibitem{Pascoal2006}
Marta M.~B. Pascoal, M.~Eug\'{e}nia~V. Captivo, and Jo\~{a}o C.~N.
  Cl\'{\i}maco.
\newblock {A comprehensive survey on the quickest path problem}.
\newblock {\em Annals of Operations Research}, 147(1):5--21, August 2006.
\newblock URL: \url{http://link.springer.com/10.1007/s10479-006-0068-x}, \href
  {http://dx.doi.org/10.1007/s10479-006-0068-x}
  {\path{doi:10.1007/s10479-006-0068-x}}.

\bibitem{Shepherd2015}
F.~Bruce Shepherd and Adrian Vetta.
\newblock {The Inapproximability of Maximum Single-Sink Unsplittable, Priority
  and Confluent Flow Problems}.
\newblock {\em arXiv:1504.0627}, 2015.
\newblock URL: \url{http://arxiv.org/abs/1504.0627}, \href
  {http://arxiv.org/abs/1504.0627} {\path{arXiv:1504.0627}}.

\bibitem{Skutella2009}
Martin Skutella.
\newblock {An introduction to network flows over time}.
\newblock In William Cook, L\'{a}szl\'{o} Lov\'{a}sz, and Jens Vygen, editors,
  {\em Research Trends in Combinatorial Optimization}, pages 451--482.
  Springer, 2009.
\newblock URL:
  \url{http://link.springer.com/chapter/10.1007/978-3-540-76796-1\_21}.

\end{thebibliography}
